\documentclass[11pt,a4paper]{article}
\usepackage{authblk}
\usepackage{fullpage}
\usepackage{amssymb,amsmath,amsthm,amsfonts}

\usepackage[utf8]{inputenc} 
\usepackage[T1]{fontenc}    
\usepackage[pagebackref]{hyperref}       
\usepackage{url}            
\usepackage{booktabs}       
\usepackage{amsfonts}       
\usepackage{nicefrac}       
\usepackage{microtype}      
\usepackage{xcolor}         
\usepackage{amsmath}
\usepackage{amssymb}
\usepackage{mathtools}
\usepackage{amsthm}
\usepackage{algorithm}
\usepackage{algorithmic}
\usepackage{amssymb,graphicx,amsthm,dsfont,bbm}
\usepackage{xspace}
\usepackage{caption}
\usepackage{subcaption}
\usepackage{natbib}

\newcommand{\myparagraph}[1]{\smallskip

  \noindent\textbf{#1}\quad}

\newcommand{\hideforspeed}[1]{#1}

\usepackage[capitalize,noabbrev]{cleveref}

\theoremstyle{plain}
\newtheorem{theorem}{Theorem}[section]
\newtheorem{proposition}[theorem]{Proposition}
\newtheorem{lemma}[theorem]{Lemma}
\newtheorem{corollary}[theorem]{Corollary}
\theoremstyle{definition}
\newtheorem{definition}[theorem]{Definition}

\newtheorem{remark}[theorem]{Remark}
\newtheorem{example}[theorem]{Example}

\usepackage[textsize=tiny]{todonotes}

\usepackage{url}
\usepackage{wrapfig}
\usepackage{graphicx}
\usepackage{mathdots}
\usepackage{color}
\usepackage{amsmath}
\usepackage{amssymb}
\usepackage{nicefrac}
\usepackage{amsthm}
\usepackage{booktabs}
\usepackage{paralist}
\usepackage{xcolor}
\usepackage{thm-restate}

\usepackage{diagbox}
\usepackage[inline]{enumitem}
\usepackage{placeins}

\usepackage{todonotes}
\usepackage{kbordermatrix}
\usepackage{pgfplots}
\usepackage{tikz}

\usepackage{thmtools} 
\usepackage{thm-restate}

\usepackage{pgfplots}
\DeclareUnicodeCharacter{2212}{−}
\usepgfplotslibrary{groupplots,dateplot}
\usetikzlibrary{patterns,shapes.arrows,backgrounds}
\pgfplotsset{compat=newest}
\pgfdeclarelayer{back}
\pgfdeclarelayer{front}
\pgfsetlayers{back,main,front}
\makeatletter
\pgfkeys{%
  /tikz/on layer/.code={
    \pgfonlayer{#1}\begingroup
    \aftergroup\endpgfonlayer
    \aftergroup\endgroup
  },
  /tikz/node on layer/.code={
    \gdef\node@@on@layer{%
      \setbox\tikz@tempbox=\hbox\bgroup\pgfonlayer{#1}\unhbox\tikz@tempbox\endpgfonlayer\egroup}
    \aftergroup\node@on@layer
  },
  /tikz/end node on layer/.code={
    \endpgfonlayer\endgroup\endgroup
  }
}
\def\node@on@layer{\aftergroup\node@@on@layer}

\newcommand{\pfnote}[1]{}
\newcommand{\nbnote}[1]{}
\newcommand{\ssnote}[1]{}

\newcommand{\prob}{\mathbb{P}}

\newcommand{\calT}{\mathcal{T}}
\newcommand{\calM}{\mathcal{M}}

\newcommand{\pref}{\succ}
\newcommand{\CP}{{{\mathrm{CP}}}}

\newcommand{\pos}{{\mathrm{pos}}}

\newcommand{\calL}{{\mathcal{L}}}

\newcommand{\reals}{{\mathbb{R}}}

\newcommand{\calD}{{\mathcal{D}}}

\newcommand{\rawPOS}{{{\mathrm{rawPOS}}}}
\newcommand{\normPOS}{{{\mathrm{nPOS}}}}

\newcommand{\EMD}{{{\mathrm{EMD}}}}

\newcommand{\ID}{{{\mathrm{ID}}}}

\newcommand{\UN}{{{\mathrm{UN}}}}
\newcommand{\AN}{{{\mathrm{AN}}}}

\newcommand{\ST}{{{\mathrm{ST}}}}

\newcommand{\freq}{{{\mathrm{freq}}}} 

\newcommand{\IC}{{\mathrm{IC}}}
\newcommand{\Mallows}{{\mathrm{Mal}}}
\newcommand{\groupsep}{{\mathrm{GS}}}

\newcommand{\Bal}{{\mathrm{Bal}}}
\newcommand{\Flat}{{\mathrm{Flat}}}

\newcommand{\Con}{{\mathrm{Con}}}
\newcommand{\Wal}{{\mathrm{Wal}}}
\newcommand{\SP}{{\mathrm{SP}}}

\newcommand{\normphi}{{{\mathrm{norm}\hbox{-}\phi}}}

\definecolor{darkgreen}{rgb}{0,0.5,0}
\definecolor{darkpink}{rgb}{0.75,0.25,0.25}
\definecolor{RED}{rgb}{1,0,0}

\newcommand{\figurecut}[1]{}
\newcommand{\cutfornow}[1]{}
\newcommand{\APPENDIX}[1]{}
\allowdisplaybreaks

\title{Expected Frequency Matrices of Elections: Computation, Geometry, and Preference Learning}

\author[1]{Niclas Boehmer}
\author[2]{Robert Bredereck}
\author[3]{Edith Elkind}
\author[4]{\\Piotr Faliszewski}
\author[4]{Stanisław Szufa}

\affil[1]{\small
  Technische Universit\"at Berlin, Algorithmics and Computational 
  Complexity\protect\\
  niclas.boehmer@tu-berlin.de}
\affil[2]{\small
  TU Clausthal, 
  robert.bredereck@tu-clausthal.de}
 \affil[3]{\small
  University of Oxford,
 elkind@cs.ox.ac.uk}
 \affil[4]{\small
  AGH University,
  \{faliszew,szufa\}@agh.edu.pl}
\date{\today}

\begin{document}
\maketitle              
\begin{abstract}
We use the ``map of elections'' approach
of Szufa et al. (AAMAS-2020) to analyze several well-known
vote distributions.
For each of them,
we give 
an explicit formula or an efficient algorithm for 
computing its frequency matrix, which captures the probability that a given 
candidate appears in a given position in a sampled vote. We use these matrices
to draw the ``skeleton map'' of distributions, evaluate its robustness, and
analyze its properties. 
Finally, we develop a general and unified framework for learning the distribution of real-world preferences using the frequency matrices of established vote distributions.
\end{abstract}

\section{Introduction}

Computational social choice is a research area at the intersection of social choice (the science of collective decision-making) and computer science, which focuses on the algorithmic analysis of problems related to preference aggregation and elicitation \citep{bra-con-end:b:comsoc}. Many of the early papers in this field were primarily theoretical, focusing on establishing the worst-case complexity of winner determination and strategic behavior under various voting rules---see, e.g., the papers of \citet{hem-hem-rot:j:dodgson}, \citet{dwo-kum-nao-siv:c:rank-aggregation}, and \citet{con-lan-san:j:when-hard-to-manipulate}---but 
more recent work often combines theoretical investigations with empirical analysis. 
For example, formal bounds on the running time and/or approximation ratio of a winner determination algorithm can be complemented by experiments 
that evaluate its performance on realistic instances; see, e.g., the works of~\citet{con:c:slater}, \citet{bet-bre-nie:j:kemeny}, \citet{fal-lac-pet-tal:c:csr-heuristics} and \citet{wan-sik-she-zha-jia-xia:c:practical-algo-stv}. 

However, performing high-quality experiments requires the ability to
organize and understand the available data. One way to achieve this is
to form a so-called ``map of elections,'' recently introduced by
\citet{szu-fal-sko-sli-tal:c:map} and extended by
\citet{boe-bre-fal-nie-szu:t:compass}. The idea is as follows. First,
we fix a distance measure between elections. Second, we sample a number of
elections from various distributions and real-life datasets---e.g.,
those collected in PrefLib~\citep{mat-wal:c:Preflib}---and measure the
pairwise distances between them. Third, we embed these elections into the
2D plane, mapping each election to a point so that the Euclidean
distances between points are approximately equal to the distances
between the respective elections.  Finally, we plot these points,
usually coloring them to indicate their origin (e.g., the
distribution from which a given election was sampled); see
\Cref{fig:compass_map} later in the paper 
for an example
of such a map. A location of an election on a map provides useful information about its properties.  For example,
\citet{szu-fal-sko-sli-tal:c:map} and
\citet{boe-bre-fal-nie:c:winner-robustness,boe-bre-fal-nie-szu:t:compass} have shown that it can be used to predict (a)~the Borda 
score of the winner of the election, (b)~the running time of ILP solvers computing the
winners under
the Harmonic-Borda multiwinner voting rule, or (c)~the robustness of
Plurality and Borda winners.  Moreover, real-world elections of the
same type (such as the ones from politics, sports, or surveys) tend to
cluster in the same areas of the map; see also the positions on the
map of the datasets collected
by \citet{boe-sch:t:collecting-data}.
As such, the map has proven to be a useful framework to analyze the nature of elections and to visualize experimental results in a non-aggregate fashion.

Unfortunately, extending the map to incorporate additional examples and
distributions is a challenging task, as the visual representation
becomes cluttered and, more importantly, the embedding algorithms,
which map elections to points in 2D, find it more difficult to
preserve pairwise distances between points as the number of points
increases. It is therefore desirable to reduce the number of points in
a way that preserves the key features of the framework.

We address this challenge by drawing a map of {\em distributions}
rather than individual elections, which we call the {\em skeleton
  map}. That is, instead of sampling 20--30 points from each
distribution and placing them all on the map, as
\citet{szu-fal-sko-sli-tal:c:map} and
\citet{boe-bre-fal-nie-szu:t:compass} do (obtaining around 800 points
in total), we create a single point for each distribution. This
approach is facilitated by the fact that prior work on the ``map of
elections'' framework represented elections by their \emph{frequency
  matrices}, which capture their essential features.
The starting point of our work is the
observation that this representation extends to distributions in a
natural way.
Thus, if we can compute the frequency matrix of some distribution $\mathcal D$, then,
instead of sampling elections from $\mathcal D$ and creating a point
on the map for each sample, we can create a single point for
$\mathcal D$ itself.

\myparagraph{Our Contribution.}
We provide three sets of results.
First, for a number of prominent vote distributions, we show how to compute their frequency matrices, 
by providing an explicit formula or an efficient algorithm.
  Second, we draw the map of distributions (the {\em skeleton map}) and argue for its credibility and robustness. 
Finally, we use our results to
estimate the parameters of the distributions that are closest to the real-world elections considered by \citet{boe-bre-fal-nie-szu:t:compass}. 
In more detail, we work in the setting of preference learning, where we are given an election and we want to learn the parameters of some distribution, so as to maximize 
the similarity of the votes sampled from this distribution and the input election. For example, we may be interested in fitting the classic
model of \citet{mal:j:mallows}. This model is parameterized by a
central vote~$v$ and a dispersion parameter~$\phi$, which specifies
how likely it is to generate a vote at some distance from the central
one (alternatively, one may use, e.g., the Plackett--Luce model).  Previous
works on preference learning typically proposed algorithms to learn the
parameters of one specific (parameterized) vote distribution (see,
e.g., the works of
\citet{lu-bou:j:sampling-mallows,MM09,DBLP:conf/uai/MeilaC10,DBLP:journals/jmlr/VitelliSCFA17,DBLP:journals/csda/MurphyM03,DBLP:conf/nips/AwasthiBSV14}
for (mixtures of) the Mallows model and the works of
\citet{DBLP:conf/icml/GuiverS09,hunter2004mm,minka2004power,gormley2008exploring}
for (mixtures of) the Plackett--Luce model).
Using frequency matrices, we offer a more general approach.  Indeed,
given an election and a parameterized vote distribution whose
frequency matrix we can compute, the task of learning the distribution's parameters boils down to finding
parameters that minimize the distance between the election and the matrices of the
distribution.  While this minimization problem may
be quite challenging, our approach offers a uniform framework for dealing
with multiple kinds of distributions at the same time. 
We find that 
for the case of the Mallows distribution, our approach learns parameters
very similar to those established using maximum likelihood-based approaches.
Omitted proofs and discussions
are in the appendix. The source code used for
the experiments is available in a GitHub repository\footnote{\url{github.com/Project-PRAGMA/Expected-Frequency-Matrices-NeurIPS-2022}}.

\section{Preliminaries}\label{se:prelims}

Given an integer $t$, we write $[t]$ to denote the set $\{1, \ldots, t\}$.
We interpret a vector $x \in \reals^{m}$ as an $m \times 1$
matrix (i.e., we use column vectors as the default).

\myparagraph{Preference Orders and Elections.}
Let $C$ be a finite, nonempty set of candidates. We refer to total
orders over~$C$ as \emph{preference orders} (or, equivalently,
\emph{votes}), and denote the set of all
preference orders over~$C$ by $\calL(C)$. Given a vote $v$ and a candidate $c$, by
$\pos_v(c)$ we mean the position of~$c$ in~$v$ (the top-ranked
candidate has position~$1$, the next one has position $2$, and so on).
If a candidate $a$ is ranked above another candidate $b$ in
vote $v$, we write $v \colon a \pref b$. 
Let $\mathrm{rev}(v)$ denote the {\em reverse} of vote $v$.
An
{\em election} $E = (C,V)$ consists of a set $C = \{c_1, \ldots, c_m\}$ of candidates and a
collection $V = (v_1, \ldots, v_n)$ of votes. Occasionally we refer to the elements of $V$ as voters rather than votes.

\myparagraph{Frequency Matrices.}
Consider an election $E = (C,V)$ with $C = \{c_1, \ldots, c_m\}$ and
$V = (v_1, \ldots, v_n)$. For each candidate $c_j$ and position
$i \in [m]$, we define $\#\freq_E(c_j,i)$ to be the fraction of the
votes from $V$ that 
rank $c_j$ in position $i$. We define the column vector
$\#\freq_E(c_j)$ to be $(\#\freq_E(c_j,1), \ldots, \#\freq_E(c_j,m))$
and matrix $\#\freq(E)$ to consist of vectors
$\#\freq_E(c_1), \ldots, \#\freq_E(c_m)$. We refer to $\#\freq(E)$ as the {\em frequency
  matrix of election~$E$}.  Frequency matrices are bistochastic, i.e.,
their entries are nonnegative and each of their rows and columns sums
up to one.
\begin{example}\label{ex:election-matrix}
  Let $E = (C,V)$ be an election with candidate set
  $C = \{a,b,c,d,e\}$ and four voters, $v_1$, $v_2$, $v_3$, and
  $v_4$. Below, we show the voters' preference orders (on the left) and the
  election's frequency matrix (on the right). \vspace{-0.8cm}
 \begin{center}
  \begin{minipage}[b]{0.2\columnwidth}
    \begin{align*}
      &v_1 \colon a \pref b \pref c \pref d \pref e, \\
      &v_2 \colon c \pref b \pref d \pref a \pref e, \\
      &v_3 \colon d \pref e \pref c \pref b \pref a, \\
      &v_4 \colon b \pref c \pref a \pref d \pref e. 
     \end{align*}
  \end{minipage}\quad\quad
  \begin{minipage}[b]{0.4\columnwidth}
    \begin{align*}
      \kbordermatrix{ & a & b & c & d & e  \\
      1 &                \nicefrac{1}{4} & \nicefrac{1}{4} & \nicefrac{1}{4} & \nicefrac{1}{4} & 0\\
      2 &                0               & \nicefrac{1}{2} & \nicefrac{1}{4} & 0               & \nicefrac{1}{4}\\
      3 &                \nicefrac{1}{4} & 0               & \nicefrac{1}{2} & \nicefrac{1}{4} & 0\\
      4 &                \nicefrac{1}{4} & \nicefrac{1}{4} & 0               & \nicefrac{1}{2} & 0\\
      5 &                \nicefrac{1}{4} & 0               & 0               & 0               & \nicefrac{3}{4}   }
    \end{align*}
  \end{minipage}
  \end{center}
\end{example}

Given a vote $v$, we write $\#\freq(v)$ to denote the frequency matrix
of the election containing this vote only; $\#\freq(v)$ is a
permutation matrix, with a single $1$ in each row and in each
column. Thus, for an election $E = (C,V)$ with $V = (v_1,\ldots, v_n)$
we have
$\textstyle \#\freq(E) = \frac{1}{n}\cdot\sum_{i=1}^n \#\freq(v_i)$.

\myparagraph{Compass Matrices.}
For even $m$, \citet{boe-bre-fal-nie-szu:t:compass} defined the following four
$m \times m$ ``compass'' matrices,
which appear to be 
extreme on the
``map of elections'':
\begin{enumerate}
\item The {\em identity matrix}, $\ID_m$, has ones on the diagonal and zeroes everywhere
  else (it corresponds to an election where all voters agree on a single
  preference order).
\item The {\em uniformity matrix}, $\UN_m$, has all entries equal to $\nicefrac{1}{m}$ (it
  corresponds to lack of agreement; each candidate is ranked at each position equally often).
\item The {\em stratification matrix}, $\ST_m$, is partitioned into
  four quadrangles, where all entries in the top-left and bottom-right
  quadrangles are equal to $\nicefrac{2}{m}$, and all other entries
  are equal to zero (it corresponds to partial agreement; the voters
  agree which half of the candidates is superior, but disagree on
  everything else).
\item The {\em antagonism matrix}, $\AN_m$, has values
  $\nicefrac{1}{2}$ on both diagonals and zeroes elsewhere (it
  captures a conflict: it is a matrix of an election where  
  half of the voters rank the candidates in one
  way and half of the voters rank them in the opposite way).
\end{enumerate}
Below, we show examples of these matrices for $m=4$:
{\footnotesize \begin{align*}
  \UN_4 = \footnotesize
        \setlength{\arraycolsep}{3pt}
	\begin{bmatrix} 
	\nicefrac{1}{4} & \nicefrac{1}{4} & \nicefrac{1}{4} &\nicefrac{1}{4} \\
	\nicefrac{1}{4} & \nicefrac{1}{4} & \nicefrac{1}{4} &\nicefrac{1}{4} \\
	\nicefrac{1}{4} & \nicefrac{1}{4} & \nicefrac{1}{4} &\nicefrac{1}{4} \\
	\nicefrac{1}{4} & \nicefrac{1}{4} & \nicefrac{1}{4} &\nicefrac{1}{4} \\
	\end{bmatrix},
  \ID_4 = \footnotesize
        \setlength{\arraycolsep}{3pt}
	\begin{bmatrix} 
	1 & 0 & 0 & 0 \\
	0 & 1 & 0 & 0 \\
	0 & 0 & 1 & 0 \\
	0 & 0 & 0 & 1 \\
	\end{bmatrix}, 
  \ST_4 = \footnotesize
        \setlength{\arraycolsep}{2.5pt}
	\begin{bmatrix}
	\nicefrac{1}{2} & \nicefrac{1}{2} & 0 & 0 \\
	\nicefrac{1}{2} & \nicefrac{1}{2} & 0 & 0 \\
	0 & 0 & \nicefrac{1}{2} & \nicefrac{1}{2} \\
	0 & 0 & \nicefrac{1}{2} & \nicefrac{1}{2} \\
	\end{bmatrix}, 
  \AN_4 = \footnotesize
        \setlength{\arraycolsep}{3pt}
  	\begin{bmatrix} 
	\nicefrac{1}{2} & 0 & 0 & \nicefrac{1}{2} \\
	0 & \nicefrac{1}{2} & \nicefrac{1}{2}  & 0 \\
	0 & \nicefrac{1}{2} & \nicefrac{1}{2}  & 0 \\
	\nicefrac{1}{2} & 0 & 0 & \nicefrac{1}{2} \\
	\end{bmatrix}.
\end{align*}}%
We omit the subscript in the names of these matrices if its value is clear from the context or irrelevant.

\myparagraph{EMD.}  Let $x = (x_1, \ldots, x_n)$ and
$y = (y_1, \ldots, y_n)$ be two vectors with nonnegative real entries
that sum up to~$1$.  Their {\em Earth mover's distance}, denoted
$\EMD(x,y)$, is the cost of transforming $x$ into $y$ using operations
of the form: Given indices $i, j \in [n]$ and a positive value
$\delta$ such that $x_i \geq \delta$, at the cost of
$\delta \cdot |i - j|$, replace $x_i$ with $x_i-\delta$ and $x_j$ with
$x_j + \delta$ (this corresponds to moving $\delta$ amount of
``earth'' from position $i$ to position~$j$). $\EMD(x,y)$ can be
computed in polynomial time by a standard greedy
algorithm. 

\myparagraph{Positionwise Distance \citep{szu-fal-sko-sli-tal:c:map}.}
Let $A = (a_1, \ldots, a_m)$ and $B = (b_1, \ldots, b_m)$ be two
$m \times m$ frequency matrices.  Their {\em raw positionwise distance} is $\rawPOS(A,B) = \min_{\sigma \in S_m}\sum_{i=1}^m \EMD(a_i, b_{\sigma(i)})$, where $S_m$ denotes the set of all permutations over
$[m]$.
We will normalize these distances by $\frac{1}{3}(m^2-1)$, which \citet{boe-bre-fal-nie-szu:t:compass,DBLP:journals/corr/abs-2205-00492} proved to be the maximum distance between two $m\times m$ frequency matrices and the distance between $\ID_m$ and $\UN_m$:  
$\normPOS(A,B) = \frac{\rawPOS(A,B)}{\frac{1}{3}(m^2-1)}$.  For two
elections $E$ and $F$ with equal-sized candidate sets, their
positionwise distance, raw or normalized, is defined as the 
positionwise distance between their frequency matrices. 

\myparagraph{Paths Between the Compass Matrices.}  Let~$X$ and~$Y$ be
two compass matrices.  \citet{boe-bre-fal-nie-szu:t:compass} showed
that if we take their affine combination $Z = \alpha X + (1-\alpha)Y$
($0 \leq \alpha \leq 1$) then
$\normPOS(X,Z) = (1-\alpha)\normPOS(X,Y)$ and
$\normPOS(Z,Y) = \alpha\normPOS(X,Y)$.  Such affine combinations form
direct paths between the compass matrices; they are also
possible between any two frequency matrices of a given size, not just the compass ones, but may
require shuffling the matrices' columns~\citep{boe-bre-fal-nie-szu:t:compass}.

\myparagraph{Structured Domains.}
We consider two classes of structured elections,
single-peaked elections~\citep{bla:b:polsci:committees-elections}, and
group-separable elections~\citep{ina:j:group-separable}.
For a 
discussion of these domains and the motivation behind them, see the original papers and the overviews by
\citet{elk-lac-pet:b:structured-preferences,elk-lac-pet:t:restricted-domains}.

Intuitively, an election is single-peaked if we can order the
candidates 
so that, as each voter considers the candidates in this order
(referred to as the \emph{societal axis}), 
his or her appreciation first increases and then decreases. The axis
may, e.g., correspond to the left-right political spectrum.

\begin{definition}
  Let $v$ be a vote over $C$ and let $\lhd$ be the societal axis over
  $C$. We say that $v$ is {\em single-peaked with respect to $\lhd$}
  if for every $t \in [|C|]$ its $t$ top-ranked candidates form an
  interval within $\lhd$.   
  An election is {\em single-peaked with
    respect to~$\lhd$} if all its votes are. An election is {\em
    single-peaked (SP)} if it is single-peaked with respect to some
  axis.
\end{definition}

Note that the election from Example~\ref{ex:election-matrix} is single-peaked with respect to the axis $a \lhd b \lhd c \lhd d \lhd e$.

We also consider {\em group-separable elections}, introduced by
\citet{ina:j:group-separable}. 
For our purposes, it will be convenient to use the
 tree-based definition of \citet{kar:j:group-separable}.  Let
$C = \{c_1, \ldots, c_m\}$ be a set of candidates, and consider a
rooted, ordered tree $\calT$ whose leaves are elements of $C$. The
\emph{frontier} of such a tree is the preference order that ranks the
candidates in the order in which they appear in the tree from left to
right.  A preference order is \emph{consistent} with a given tree if
it can be obtained as its frontier by reversing the order in which the
children of some nodes appear.

\begin{definition}\label{def:gs}
  An election $E = (C,V)$ is \emph{group-separable} if there is a
  rooted, ordered tree $\calT$ whose leaves are members of $C$, such
  that each vote in $V$ is consistent with~$\calT$.
\end{definition}

\noindent The trees from Definition~\ref{def:gs} form a subclass of
\emph{clone decomposition trees}, which 
are examples of PQ-trees
\citep{elk-fal-sli:c:decloning,boo-lue:j:consecutive-ones-property}.

\begin{example}
  Consider candidate set $C = \{a, b,c,d\}$, trees $\calT_1$,
  $\calT_2$, and $\calT_3$ from Figure~\ref{fig:trees}, and votes
  $v_1 \colon a \pref b \pref c \pref d$,
  $v_2 \colon c \pref d \pref b \pref a$, and
  $v_3 \colon b \pref d \pref c \pref a$.  Vote $v_1$ is consistent
  with each of the trees,
  $v_2$ is consistent with $\calT_2$ (reverse the children of $y_1$
  and $y_2$), and $v_3$ is consistent with $\calT_3$ (reverse the
  children of $x_1$ and $x_3$).
\end{example}

\section{Frequency Matrices for Vote Distributions}\label{sec:freq_dists}
We show how to compute 
frequency matrices
for several well-known distributions over votes.

\subsection{Setup and Interpretation}\label{sec:setup}

A {\em vote distribution} for a candidate set~$C$ is a function
$\calD$ that assigns a probability to each preference order over~$C$.
Formally, we require that for each $v \in \calL(C)$ it holds that
$\calD(v) \geq 0$ and $\sum_{v \in \calL(C)}\calD(v) = 1$. We say that
a vote $v$ is in the {\em support} of $\calD$ if $\calD(v)>0$. Given
such a distribution, we can form an election by repeatedly drawing
votes according to the specified probabilities.
For example, we can sample each element of $\calL(C)$ with equal probability; this distribution, which is known as
\emph{impartial culture (IC)}, is denoted
by $\calD_\IC$ (we omit the
candidate set from our notation as it will always be clear from the
context).
The {\em frequency matrix of a vote distribution $\calD$} over a candidate set $C$ is
$\#\freq(\calD) = \sum_{v \in \calL(C)} \calD(v) \cdot \#\freq(v)$.
For example, we have $\#\freq(\calD_\IC) = \UN$.  One interpretation
of $\#\freq(\calD)$ is that the entry for a candidate~$c_j$ and a
position~$i$ is the probability that a vote $v$ sampled from
 $\calD$ has~$c_j$ in position~$i$ (which we denote as
$\prob[\pos_v(c_j) = i]$).  Another interpretation is that if we sample
a large number of votes then the resulting election's frequency matrix
would be close to $\#\freq(\calD)$ with high probability. More
formally, if we let $\calM_n$ be a random variable equal to the
frequency matrix of an $n$-voter election generated according
to~$\calD$, then it holds that
$\lim_{n\rightarrow \infty} \mathbb{E}(\calM_n) = \#\freq(\calD)$.

\subsection{Group-Separable Elections}
\label{sec:gs-mat}

We first consider
sampling group-separable votes.  Given a rooted
tree $\calT$ whose leaves are labeled by elements of $C=\{c_1, \dots, c_m\}$,
let $\calD_\groupsep^\calT$ be the distribution assigning equal
probability to all votes consistent with $\calT$, 
and zero probability to all other votes; 
one can think of $\calD_\groupsep^\calT$ as impartial culture
restricted to the group-separable subdomain defined by $\calT$. To sample from $\calD_\groupsep^\calT$, we can toss a fair coin for each internal node of $\calT$, reversing the order of its children if this coin comes up heads, and output
the frontier of the resulting tree.  We focus on 
the following types of trees:

\begin{enumerate}
\item $\Flat(c_1, \ldots, c_m)$ is a tree with a single internal node,
  whose children, from left to right, are $c_1, c_2, \ldots,
  c_m$. There are only two preference orders consistent with this
  tree, $c_1 \pref \cdots \pref c_m$ and its reverse.
  
\item $\Bal(c_1, \ldots, c_m)$ is a perfectly balanced binary tree
  with frontier $c_1, \ldots, c_m$ (hence we assume the number $m$ of
  candidates to be a power of two).

\item $\CP(c_1, \ldots, c_m)$ is a binary caterpillar tree: it has
  internal nodes $x_1, \ldots, x_{m-1}$; for each $j \in [m-2]$, $x_j$
  has $c_j$ as the left child and $x_{j+1}$ as the right one, whereas
  $x_{m-1}$ has both $c_{m-1}$ and $c_m$ as children.
\end{enumerate}
The first tree in Figure~\ref{fig:trees} is flat, the second one is
balanced, and the third one is a caterpillar tree.
If $\calT$ is a caterpillar tree, then we refer to
$\calD_\groupsep^\calT$ as the {\em GS/caterpillar distribution}. We use a similar terminology for the other trees.

\begin{theorem}\label{thm:gs-tress-matrix}
  Let $F$ be the frequency matrix of
  distribution $\calD_\groupsep^\calT$. If $\calT$ is flat then
  $F = \AN$, and if it is balanced then $F = \UN$. If $\calT$ is a
  caterpillar tree $\CP(c_1, \ldots, c_m)$, then for each candidate
  $c_j$ the probability 
  that $c_j$ appears in a position $i \in [m]$ in a random vote $v$
  sampled from $\calD_\groupsep^\calT$ is: 
  \[ \textstyle \frac{1}{2^j}{ j-1 \choose i-1}\cdot {\mathbbm 1}_{i \leq j}
 + \frac{1}{2^j}{ j-1 \choose (i-1) - (m-j)}\cdot {\mathbbm 1}_{i > m-j}.
 \]
\end{theorem}

\begin{proof}
  The cases of flat and balanced trees are immediate, so we focus on
  caterpillar trees. Let $\calT = \CP(c_1, \ldots, c_m)$
  with internal nodes $x_1, \dots, x_{m-1}$, and
  consider a candidate~$c_j$ and a position $i \in [m]$. 
  Let $v$ be a random variable equal to a vote sampled from
  $\calD_\groupsep^\calT$.  
  We say that a node
  $x_\ell$, $\ell\in [m-1]$, is {\em reversed} if the order of its children is reversed.
  Note that for $\ell< r$ it holds that~$c_r$ precedes~$c_\ell$ in the frontier if and only if $x_\ell$ is reversed.
  Suppose that $x_j$ is not reversed. Then $v$ ranks $c_j$ above each of $c_{j+1}, \ldots, c_m$. This means that for $c_j$ to be ranked
  exactly in position $i$, it must be that $j \geq i$ and exactly
  $i-1$ nodes among $x_1, \ldots, x_{j-1}$ are not reversed. If $j \geq i$, the probability that $x_j$ and $i-1$ nodes among $x_1, \ldots, x_{j-1}$ are not reversed is $\frac{1}{2^j}\cdot { j-1 \choose i-1}$. 
  On the other hand, if $x_j$
  is reversed, then $v$ ranks candidates $c_{j+1}, \ldots, c_m$ above
  $c_j$. As there are $m-j$ of them, for $c_j$ to be ranked exactly
  in position $i$ it must hold that $i > m-j$ and exactly
  $(i-1) - (m-j)$ nodes among $x_1, \ldots, x_{j-1}$ are not
  reversed. This happens with probability
  $\frac{1}{2^j}\cdot { j-1 \choose (i-1) - (m-j)}$.
\end{proof}
  
Regarding distributions $\calD_\groupsep^\calT$ not handled in
Theorem~\ref{thm:gs-tress-matrix}, we still can compute their
frequence matrices efficiently.

\begin{restatable}{theorem}{gsdp}\label{thm:gs-dp}
  There is an algorithm that given a tree $\calT$ 
  computes $\#\freq(D_\groupsep^\calT)$ using polynomially many arithmetic
  operations with respect to the number of nodes in~$\calT$.
\end{restatable}
\begin{proof}
  Let $\calT$ be the input tree, and let $C = \{c_1, \ldots, c_m\}$ be
  its candidate set.  We will give an algorithm for computing the
  probability that a given candidate $c_j$ appears in position
  $i \in [m]$ in a vote sampled from
  $\calD_\groupsep^\calT$.

  Let $x$ be some node of $\calT$ (either internal or a leaf). Let $\calT_x$ be the tree obtained from $\calT$ by deleting all descendants of $x$, so that $x$ becomes a leaf, 
  and for each subset $S$ of internal nodes of $\calT_x$, let $\calT_x^S$ be the ordered tree obtained by starting with $\calT_x$ and reversing the nodes in the set $S$. For each $t \in \{0\} \cup [m-1]$ we define $f(x,t)$ to be the probability that if we reversed each internal node of $\calT_x$ with probability $\nicefrac{1}{2}$ then~$x$ would be preceded by exactly $t$ candidates in the frontier of the resulting ordered tree. We compute $f(x,t)$ using dynamic programming.

  Let $\mathit{root}$ be the root of $\calT$. Then
  $f(\mathit{root},0) = 1$, and $f(\mathit{root}, t)=0$ for $t\in [1, m-1]$.
  Next, let $x$ be some node of $\calT$
  other than the root, let $p$ be the parent of $x$, and let $\ell$
  and $r$ be the number of leaves that are descendants of 
  $x$'s left siblings and $x$'s right siblings in $\calT$,
  respectively. We claim that for each $t \in \{0\} \cup [m-1]$ we have:
  \[
    \textstyle f(x,t) = \frac{1}{2} f(p,t-\ell) + \frac{1}{2}f(p,t-r).
  \]
  To see why this formula is correct, observe that 
  if $p\not\in S$ then $x$ appears in position $t$ in the frontier of $\calT_x^S$ if and only if $p$ appears in position $t-\ell$ in the 
  frontier of $\calT^S_p$: indeed, in the frontier of $\calT_x^S$ the node $x$ appears after all predecessors of $p$ in the frontier of $\calT_p^S$
  as well as after the $\ell$ leaves that are the descendants of $x$'s left siblings in $\calT$. Similarly, if $p\in S$ then $x$ appears in position $t$ in the frontier of $\calT_x^S$ if and only if $p$ appears in position $t-r$ in the 
  frontier of $\calT_p^{S\setminus\{p\}}$: indeed, in the frontier of $\calT_x^S$ the node $x$ appears after all predecessors of $p$ in the frontier of $\calT_p^{S\setminus\{p\}}$
  as well as after the $r$ leaves that are the descendants of $x$'s right siblings in $\calT$. Since $p$ is reversed with probability $\frac12$, the recurrence follows.
  
  The above formula and standard dynamic programming allow us to compute
  all the values of $f$ using $O(m^2)$ arithmetic operations (note
  that there are at most $O(m)$ internal nodes). To complete the proof, observe that the probability that candidate $c_j$ ends up in position $i$ is $f(c_j,i-1)$.
\end{proof}

\subsection{From Caterpillars to Single-Peaked Preferences.}
\label{sec:gs-vs-sp}
There is a relationship between
GS/caterpillar votes and
single-peaked ones, which will
be very useful when computing one of the frequency matrices
in the next section.
To prove it, we use the following lemma, which is implicit in the work of \citet{FKO20}.

\begin{lemma}\label{lem:gs}
Let $\calT=\CP(c_1, \dots, c_m)$.
A ranking $v$ over $\{c_1, \dots, c_m\}$ belongs to the support of $\calD_\groupsep^\calT$
if and only if there exists a subset $C'\subseteq C\setminus\{c_m\}$
such that in $v$:
\begin{enumerate}
\item
  (1)~all alternatives in $C'$ are ranked above $c_m$ and all
  alternatives in $(C\setminus\{c_m\})\setminus C'$ are ranked below
  $c_m$;
\item
  (2)~for all $c_i, c_j\in C'$ with $i<j$ the alternative $c_i$ is
  ranked above $c_j$;
\item
  (3) for all $c_i, c_j\not\in C'\cup \{c_m\}$ with $i<j$ the
  alternative $c_i$ is ranked below $c_j$.
\end{enumerate}
\end{lemma}
That is, in $v$ the alternatives in $C'$ appear in the increasing order of indices, followed by $c_m$, followed by the remaining alternatives in the decreasing order of indices,
i.e., the sequence of indices in $v$ is ``single-peaked''.
Using this observation, we
establish a bijection between the votes in the support of
$\calD_\groupsep^\calT$ and single-peaked votes.

\begin{restatable}{theorem}{gstosp}\label{thm:gs-to-sp} 
Given a ranking $v$ over $C=\{c_1, \dots, c_m\}$, let $\widehat{v}$
be another ranking over $C$ such that, for each $j\in [m]$, if $c_j$ is ranked in position $i$ in $v$ then 
$c_i$ is ranked in position $m-j+1$ in $\widehat{v}$.
Suppose that $v$ is in the support of $\calD_\groupsep^\calT$, 
where $\calT=\CP(c_1, \dots, c_m)$. Then
$\widehat{v}$ is single-peaked with respect to $c_1\lhd\dots\lhd c_m$. 
\end{restatable}
\begin{proof}
Suppose that $c_m$ is ranked in position $z$ in $v$. Then $c_z$ is ranked first in $\widehat{v}$. 

Consider two candidates $c_x, c_y$ with $x<y<z$. We will prove that in $\widehat{v}$ candidate $c_y$ is ranked above $c_x$, i.e.,  $\pos_{\widehat{v}}(c_x) > \pos_{\widehat{v}}(c_y)$. 
Let $k = \pos_{\widehat{v}}(c_x)$, $\ell= \pos_{\widehat{v}}(c_y)$.
Then, in $v$, alternative $c_{m-k+1}$ is ranked in position $x$ and alternative $c_{m-\ell+1}$ is ranked in position $y$. As we have $x<y<z$ and $v$
is sampled from $\calD_\groupsep^\calT$, 
by Lemma~\ref{lem:gs} we have $m-k+1 < m-\ell+1$, and hence $k>\ell$. 

Similarly, if we have two alternatives
$c_{x}, c_{y}$ with $z<y<x$, we can show that
$\pos_{\widehat{v}}(c_x) > \pos_{\widehat{v}}(c_y)$.
Thus, $\widehat{v}$ is single-peaked 
with respect to $c_1\lhd\dots\lhd c_m$, as claimed.
\end{proof}

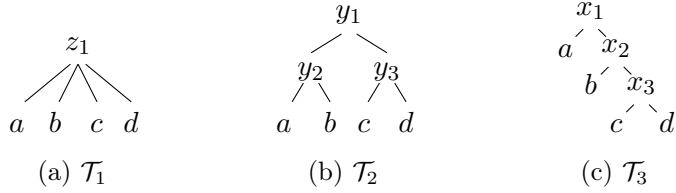
\begin{figure}
  \newcommand{\myscale}{1}
  \centering
  \begin{subfigure}[b]{0.2\columnwidth}
    \centering
    \begin{tikzpicture}[scale=\myscale]
      \node[anchor=south] at (0,0) {$z_1$};
      \draw (-0.1,0) -- (-0.7,-0.5);
      \draw (0,0) -- (-0.25,-0.5);
      \draw (0,0) -- ( 0.25,-0.5);
      \draw (0.1,0) -- (0.7,-0.5);
      \node[anchor=north] at (-0.7,-0.5) {$a$\phantom{d}};
      \node[anchor=north] at (-0.3,-0.5) {$b$};
      \node[anchor=north] at ( 0.35,-0.5) {$c$\phantom{d}};
      \node[anchor=north] at (0.7,-0.5) {$d$};
    \end{tikzpicture}
    \caption{$\calT_1$}
  \end{subfigure}\quad
  \begin{subfigure}[b]{0.2\columnwidth}
    \centering
    \begin{tikzpicture}[scale=\myscale]
      \node[anchor=south] at (0,0) {$y_1$};
      \draw (-0.1,0) -- (-0.5,-0.25);
      \draw ( 0.1,0) -- ( 0.5,-0.25);
      \node[anchor=north] at (-0.5,-0.25) {$y_2$};
      \node[anchor=north] at (0.5,-0.25) {$y_3$};
      \draw (-0.6,-0.65) -- (-0.75,-0.9);
      \node[anchor=north] at (-0.75,-0.9) {$a$\phantom{d}};
      \draw (-0.4,-0.65) -- (-0.25,-0.9);
      \node[anchor=north] at (-0.25,-0.9) {$b$};
      \draw (0.6,-0.65) -- (0.75,-0.9);
      \node[anchor=north] at (0.75,-0.9) {$d$};
      \draw (0.4,-0.65) -- (0.25,-0.9);
      \node[anchor=north] at (0.3,-0.9) {$c$\phantom{d}};
    \end{tikzpicture}
    \caption{$\calT_2$}
  \end{subfigure}\quad
  \begin{subfigure}[b]{0.2\columnwidth}
  \centering
    \begin{tikzpicture}[scale=\myscale]
      \foreach \i/\n  in {1/a,2/b,3/c} {
        \node[anchor=south] at (0+\i/3-1,0-\i/2-1) {$x_\i$};
        \draw (0+\i/3-1+0.1,0-\i/2-1) -- (0+\i/3+0.1-1+0.1,0-\i/2-1-0.1);
        \draw (0+\i/3-1-0.1,0-\i/2-1) -- (0+\i/3-0.1-1-0.1,0-\i/2-1-0.1);
        \node[anchor=south] at (0+\i/3-1/3-1,0-\i/2-1/2-1) {$\n$};
      }
      \renewcommand{\i}{3}
      \node[anchor=south] at (0+\i/3+1/3-1,0-\i/2-1/2-1) {$d$};
    \end{tikzpicture}
    \caption{$\calT_3$}
  \end{subfigure}
  \caption{\label{fig:trees}Three examples of clone decomposition trees.}
\end{figure}

There are exactly $2^{m-1}$ votes in the support of $\calD_\groupsep^\calT$ (this follows by simple counting)
and there are $2^{m-1}$ votes that are single-peaked with respect to $c_1\lhd\dots\lhd c_m$. 
As $u\neq v$ implies $\widehat{u}\neq \widehat{v}$, it follows that the mapping $v\mapsto\widehat{v}$ is a bijection between all votes in the support of $\calD_\groupsep^\calT$ and all votes that are 
single-peaked with respect to $c_1\lhd\dots\lhd c_m$.

\subsection{Single-Peaked 
Elections}
\label{sec:sp-mat}
We consider two models of generating single-peaked elections, one due
to \citet{wal:t:generate-sp} and one due to
\citet{con:j:eliciting-singlepeaked}. Let us fix a candidate set
$C = \{c_1, \ldots, c_m\}$ and a societal axis
$c_1 \lhd \cdots \lhd c_m$.  Under the Walsh distribution, denoted
$\calD_{\SP}^{\Wal}$, each vote that is single-peaked with respect to $\lhd$
has equal probability (namely, $\frac{1}{2^{m-1}}$), 
and all other votes have probability zero.
By Theorems~\ref{thm:gs-tress-matrix} and~\ref{thm:gs-to-sp}, 
we immediately obtain the frequency matrix
for the Walsh distribution (in short, it is the transposed
matrix of the GS/caterpillar distribution).
\begin{corollary}
  Consider a candidate set $C = \{c_1, \ldots, c_m\}$ and an axis
  $c_1 \lhd \cdots \lhd c_m$. The probability that candidate $c_j$ appears in position $i$ in a vote sampled from $\calD_{\SP}^{\Wal}$ is: $\frac{1}{2^{m-i+1}}{ m-i \choose j-1}\cdot {\mathbbm 1}_{j \le m-i+1}
    + \textstyle\frac{1}{2^{m-i+1}}{ m-i \choose j-i}\cdot {\mathbbm 1}_{j > i-1}$.
\end{corollary}

To sample a vote from the Conitzer distribution, $\calD_\SP^\Con$ (also known as the {\em random peak distribution}),
we pick some candidate
$c_j$ uniformly at random and rank him or her on top. Then we perform
$m-1$ iterations, where in each we choose (uniformly at random)
a candidate directly to the right or the left of the
already selected ones, and place him or her in the highest available
position in the vote.

\begin{restatable}{theorem}{conmain}
Let $c_1 \lhd \cdots \lhd c_m$ be the societal axis, where
  $m$ is an even number, and let $v$ be a random vote sampled from
  $\calD^\Con_\SP$ for this axis. For $j \in [\frac{m}{2}]$ and
  $i \in [m]$ we have:
  \[
    \prob[ \pos_v(c_j) = i ] = \begin{cases}
      \nicefrac{2}{2m}      & \text{if  $i < j$,} \\
      \nicefrac{(j+1)}{2m}   & \text{if $i = j$,} \\
      \nicefrac{1}{2m}     & \text{if $j < i < m - j + 1$,} \\
      \nicefrac{(m-j+1)}{2m} & \text{if $i = m-j+1$,} \\
      0 & \text{if $i+j > m$.}
    \end{cases}
  \]
  Further, for each candidate $c_j \in C$ and each position
  $i \in [m]$ we have
  $\prob[ \pos_v(c_j) = i] = \prob[ \pos_v(c_{m-j+1}) = i]$.
\end{restatable}

\subsection{Mallows Model}\label{sec:Mallows}
Finally, we consider the classic Mallows distribution. It has two
parameters, a central vote $v^*$ over~$m$ candidates and a dispersion
parameter $\phi \in [0,1]$. The probability of sampling a vote $v$
from this distribution (denoted $\calD_\Mallows^{v^*,\phi}$) is:
$\calD_\Mallows^{v^*,\phi}(v)=\frac{1}{Z}\phi^{\kappa(v,v^*)},$
where
$Z=1\cdot (1+\phi)\cdot (1+\phi+\phi^2)\cdot \dots
\cdot(1+\dots+\phi^{m-1})$ is a normalizing constant and
$\kappa(v,v^*)$ is the swap distance between $v$ and $v^*$ (i.e., the
number of swaps of adjacent candidates needed to transform $v$
into $v^*$).
In our experiments, we 
consider a new
parameterization,
introduced by
\citet{boe-bre-fal-nie-szu:t:compass}.  It
uses a \emph{normalized dispersion parameter} $\normphi$, which is
converted to a value of $\phi$ so that the expected swap distance
between the central vote $v^*$ and a sampled vote $v$ is
$\frac{\normphi}{2}$ times the maximum swap distance between two votes
(so, $\normphi=1$ is equivalent to $\IC$ and for $\normphi=0.5$
we get elections that lie close to the middle of the $\UN$--$\ID$
path).  

Our goal is now to compute the frequency matrix of $\calD_\Mallows^{v^*,\phi}$. That is, given the
candidate ranked in position $j$ in the central vote, 
we want
to compute the probability that he or she appears in a given position
$i \in [m]$ in the sampled vote.
Given a positive integer $m$, consider the candidate set
$C(m) = \{c_1, \ldots, c_m\}$ and the central vote
$v^*_m\colon c_1 \pref \cdots \pref c_m$.  Fix a candidate $c_j \in C(m)$, and a position $i \in [m]$. For every integer $k$ between $0$ and $\nicefrac{m(m-1)}{2}$, let
$S(m,k)$ be the number of votes in $\calL(C(m))$ that are at swap
distance $k$ from $v^*_m$, and define $T(m,k,j,i)$ to be the number
of such votes that have $c_j$ in position $i$.  One can compute
$S(m,k)$ in time polynomial in $m$ \citep{oeis}; using $S(m,k)$,
we 
show that the same holds for $T(m,k,j,i)$. 
\begin{restatable}{lemma}{lmallows}\label{lem:mallows}
There is an algorithm that computes $T(m,k,j,i)$ in polynomial time
  with respect to $m$.
\end{restatable}
\begin{proof}
  Our algorithm is based on dynamic programming. Fix
  some $m > 0$, $k \in [\nicefrac{m(m-1)}{2}]$, and $j, i \in [m]$.
  We claim that:
  \[
    T(m,k,m,i) = S(m-1,k-(m-i)).
  \]
  Indeed, let $v$ be a vote over $C(m)$ that ranks $c_m$ in position $i$, and let $v'$ be its restriction to $\{c_1, \dots, c_{m-1}\}$. Then 
  $v$ can be obtained from $v'$ by inserting $c_m$ right behind the candidate in position $i-1$. If $v'$ is at swap distance $k'$ from
  $v^*_{m-1}$, then the resulting vote is at swap distance $k'+ (m-i)$ from $v^*_m$, since $c_m$ contributes $m-i$ additional swaps. Thus, $T(m, k, m, i)$ is equal to the number of votes
  in $\calL(C(m-1))$ at swap distance $k-(m-i)$ from $v^*_{m-1}$.

  Next, we claim that for each $j < m$, we have:
  \begin{align*}
    \textstyle
    T(m,k,j,i) &= 
    \textstyle \sum_{\ell=i+1}^{m} T[m-1,k-(m-\ell),j,i]  \\
    &+ \textstyle \sum_{\ell=1\phantom{+1}}^{i-1} T[m-1,k-(m-\ell),j,i-1].
  \end{align*}
  To see why this holds, again consider inserting $c_m$ at some
  position in a vote $v'$ over $\{c_1, \dots c_{m-1}\}$. Candidate $c_j$ will end up in position $i$ in the resulting vote if (1)
  $c_j$ was in position $i$ in $v'$ and $c_m$ was inserted after $c_j$, or if (2) $c_j$ was in
  position $i-1$ in $v'$ and $c_m$ was inserted ahead of $c_j$.
  Considering all positions in which $c_m$ can be inserted, 
  we obtain the above equality.

  Using the above equalities and the fact that $T(1,0,1,1) = 1$ (as there
  is just a single vote over $C(1)$), we can compute $T(m,k,j,i)$ by dynamic programming; our algorithm runs in polynomial time with respect to $m$.
\end{proof}

We can now express the probability of
sampling a vote $v$, where the candidate ranked in position $j$
in the central vote $v^*$ ends up in position $i$ under $\calD_\Mallows^{v^*, \phi}$, as:
\begin{equation}\label{eq:fm}
     f_m(\phi,j,i) = \textstyle \frac{1}{Z}\sum_{k=0}^{\nicefrac{m(m-1)}{2}} T(m,k,j,i) \phi^k.
\end{equation}
The correctness follows from the definitions of $T$
and $\calD_\Mallows^{v^*, \phi}$.  By Lemma~\ref{lem:mallows}, we have the following.
\begin{theorem}
There exists an algorithm that, given a number $m$ of candidates,
a vote $v^*$, and a parameter $\phi$,  
computes the frequency matrix of $\calD_\Mallows^{v^*, \phi}$
using polynomially many operations in~$m$.
\end{theorem} 
Note that \Cref{eq:fm} only depends on $\phi$, $j$ and $i$ (and,
naturally, on $m$). Using this fact, we can also compute frequency
matrices for several variants of the Mallows distribution.

\begin{remark}\label{rem:mallows-mixture}
  Given a vote $v$, two dispersion parameters $\phi$ and $\psi$,
  and a probability $p \in [0,1]$, we define the 
  distribution $p$-$\calD_\Mallows^{v, \phi, \psi}$
  as $p\cdot\calD_\Mallows^{v,\phi}+
  (1-p)\cdot\calD_\Mallows^{\mathrm{rev}(v),\psi}$,
  i.e.,  with probability $p$ we sample a vote
  from $\calD_\Mallows^{v,\phi}$ and with probability $1-p$ we
  sample a vote from $\calD_\Mallows^{\mathrm{rev}(v),\psi}$.
  The probability
  that candidate $c_j$ appears in position $i$ in the resulting vote
  is
    $p\cdot f_m(\phi,j,i) + (1-p) \cdot f_m(\psi,m-j+1,i)$.
\end{remark}

\begin{remark}\label{rem:mallowsification}
  Consider a candidate set $C =\{c_1, \ldots, c_m\}$.
  Given a vote distribution $\calD$ over $\calL(C)$ and a parameter~$\phi$, 
  define a new distribution $\calD'$ as follows: Draw a vote $\hat{v}$
  according to $\calD$ and then output a vote $v$ sampled from
  $\calD_\Mallows^{\hat{v},\phi}$; indeed, such models are quite
  natural, see, e.g., the work of
  \citet{ken-kim:c:probabilistic-winner}. For each $t \in [m]$, let
  $g(j,t)$ be the probability that $c_j$ appears in position $t$ in a
  vote sampled from~$\calD$. The probability that $c_j$ appears in
  position $i \in [m]$ in a vote sampled from $\calD'$ is
  $\sum_{t=1}^{m} g(j,t)\cdot f(\phi,t,i)$. In terms of matrix
  multiplication, this means that
    $\#\freq(\calD') = \#\freq(\calD_\Mallows^{v^*,\phi}) \cdot \#\freq(\calD)$,
  where $v^*$ is $c_1 \pref \cdots \pref c_m$. We write $\phi$-Conitzer ($\phi$-Walsh) to refer to this model where we use the Conitzer (Walsh) distribution as the underlying one
  and normalized dispersion parameter $\phi$.
\end{remark}

\section{Skeleton Map}\label{sec:experiments}\label{sec:robustness}

Our goal in this section is to form what we call a
\emph{skeleton map of vote distributions} (skeleton map, for short),
evaluate its quality and robustness, and compare it to the map of
\citet{boe-bre-fal-nie-szu:t:compass}.
Throughout this
section, whenever we speak of a distance between elections or
matrices, we mean the positionwise distance (occasionally we will also
refer to the Euclidean distances on our maps, but we will always make
this explicit).
Let $\Phi = \{0, 0.05, 0.1, \ldots, 1\}$ be a set of normalized
dispersion parameters that we will be using for Mallows-based
distributions in this section.

We form the skeleton map following the general approach of
\citet{szu-fal-sko-sli-tal:c:map} and
\citet{boe-bre-fal-nie-szu:t:compass}.
For a given number of candidates, we consider the four
compass matrices (UN, ID, AN, ST) and paths between each matrix pair consisting of their convex combinations (gray dots),
 the frequency matrices
of the Mallows distribution with normalized dispersion parameters from $\Phi$ (blue triangles), and the frequency matrices of the Conitzer (CON), Walsh (WAL), and GS/caterpillar distribution (CAT). Moreover, we add the frequency matrices of the following 
vote distributions (we again use the dispersion parameters from $\Phi$):
(i)~The distribution $\nicefrac{1}{2}$-$\calD_\Mallows^{v, \phi, \phi}$ as
  defined in \Cref{rem:mallows-mixture} (red triangles),
(ii)~the distribution where with equal probability we mix the standard
  Mallows distribution and $\nicefrac{1}{2}$-$\calD_\Mallows^{v, \phi, \phi}$ (green triangles), and
(iii)~the $\phi$-Conitzer and $\phi$-Walsh distributions 
as defined in \Cref{rem:mallowsification} (magenta and orange crosses).
For each pair of these matrices we compute their positionwise distance. Then we find an embedding of the matrices
 into a 2D plane, so that each
matrix is a point and the Euclidean distances between these points are
as similar to the positionwise distances as
possible (we use the MDS algorithm, as implemented in the Python sklearn.manifold.MDS package). In Figure~\ref{fig:skeleton_map} we show our map for the
case of $10$ candidates (the lines between some points/matrices show
their positionwise distances; 
to maintain clarity, we only provide some of them).

\begin{figure*}[t]
	\centering
 	\begin{minipage}{.32\textwidth}
 	    \vspace{0.75cm}
		\centering
 		\includegraphics[width=1.02\textwidth]{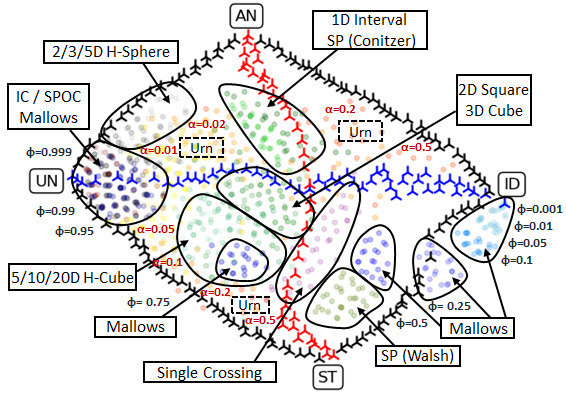}
 		\caption{Map of elections, together with the compass matrices, as presented by \citet{boe-bre-fal-nie-szu:t:compass}.}\label{fig:compass_map}
 	\end{minipage}\hfill
 	\begin{minipage}{.32\textwidth}
 	    \vspace{0.9cm}
 		\centering
 		\hideforspeed{\scalebox{0.75}{\input{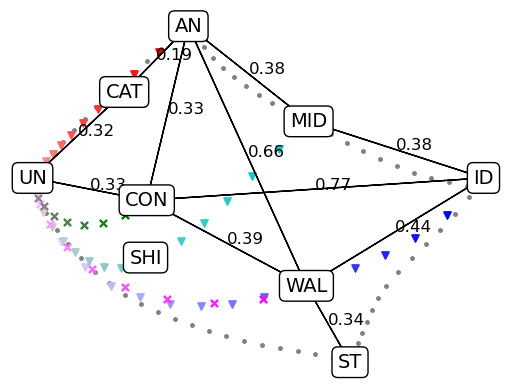}}}        \caption{The skeleton map with 10 candidates. We have $\textrm{MID} = \nicefrac{1}{2}\AN + \nicefrac{1}{2}\ID$. Each point labeled with a number is a real-world election as described in \Cref{sub:framework}.}
      \label{fig:skeleton_map}
 	\end{minipage}\hfill
 	\begin{minipage}{.32\textwidth}
 		\centering
        \includegraphics[width=0.85\textwidth]{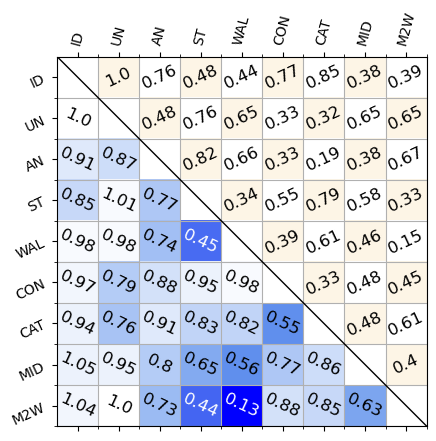}
        \caption{In the top-right part, we show the normalized positionwise distance. In the bottom-left one, we show the embedding misrepresentation ratios.}
        \label{fig:distance-distortion}
 	\end{minipage}\hfill
\end{figure*}

We now verify the credibility of the skeleton map.
As the map 
does not have many points,  
we expect its
embedding to truly reflect the positionwise distances between the
matrices. This, indeed, seems to be the case, although 
some distances are represented (much) more accurately than the others.
In Figure~\ref{fig:distance-distortion} 
we provide the following
data for a number of matrices (for $m=10$; matrix M2W is the Mallows matrix in our data set that is closest to the Walsh matrix).
  In the top-right part (the white-orange area), we give
  positionwise distances between the 
  matrices, and
  in the bottom-left part (the blue area), for each pair of
  matrices $X$ and $Y$ we report the \emph{misrepresentation ratio}
  $\frac{\mathrm{Euc}(X,Y)}{\normPOS(X,Y)}$,  
  where $\mathrm{Euc}(X,Y)$ is the Euclidean
  distance between $X$ and $Y$ in the embedding, normalized by the
  Euclidean distance between $\ID$ and $\UN$. 
The closer they are to $1$, the more accurate
is the embedding.
The misrepresentation ratios are typically between $0.8$ and
$1.15$, with many of them between $0.9$ and $1.05$. Thus, in
most cases, the map is quite accurate and offers good intuition
about the relations between the matrices.  Yet, some distances are
represented particularly badly.  As an extreme example, 
the Euclidean distance between the Walsh matrix and
the closest Mallows matrix, M2W, is off by almost a factor of $8$ (these matrices are close, but not as close as the map
suggests).  
Thus, while one always has to verify claims
suggested by the skeleton map, we view it as quite credible.
This conclusion is particularly valuable when we compare the skeleton
map and the map of \citet{boe-bre-fal-nie-szu:t:compass}, shown in
Figure~\ref{fig:compass_map}. The two maps are similar, and analogous
points (mostly) appear in analogous positions. Perhaps the biggest
difference is the location of the Conitzer matrix on the skeleton
map and Conitzer elections in the map of
\citeauthor{boe-bre-fal-nie-szu:t:compass}, but even this difference is
not huge.
We remark that the Conitzer matrix is closer to $\UN$
and $\AN$ than to $\ID$ and $\ST$, whereas for the Walsh matrix the opposite is true.
\citet{boe-bre-fal-nie-szu:t:compass} make a similar observation;
our results  allow us to make this claim formal.
In \Cref{app:robustness}, we analyze the robustness of the skeleton map with respect to varying the number of candidates. 
We find that except for pairs including the Walsh or GS/caterpillar matrices, which "travel" on the map as the number of candidates increases, the distance between each pair of matrices in the skeleton map stays
nearly constant.

\section{Learning Vote Distributions}\label{sub:framework}
We demonstrate how the positionwise distance and frequency matrices can be used to fit vote distributions to given real-world elections. 
Specifically, we consider the Mallows model ($\calD_\Mallows^{v, \phi}$) and the $\phi$-Conitzer and $\phi$-Walsh models. 
Naturally, we could
use more distributions, but we focus on showcasing the technique and the general unified approach.
Concerning our results, among others,  we verify that for Mallows model our approach is strongly correlated with existing maximum-likelihood approaches. 
Moreover, unlike in previous works, we compare the capabilities of different distributions to fit the given elections.
We remark that if we do not have an algorithm for computing a frequency
matrix of a given vote distribution, we can obtain an approximate matrix
by sampling sufficiently many votes from this distribution. 
In principle, it is also possible to deal with distributions over elections
that do not correspond to vote distributions and hence are not captured by expected frequency matrices (as
is the case, e.g., for the Euclidean models where candidates do not
have fixed positions; see the work of \citet{szu-fal-sko-sli-tal:c:map} for
examples of such models in the context of the map of elections):
If we want to compute the distance of such a distribution, 
we sample sufficiently many elections and compute their average distance from the input one. 
However, 
it remains unclear
how robust this approach is.

\myparagraph{Approach.}
To fit our vote distributions to a given election, we compute the election's distance to the frequency matrices of 
$\calD_\Mallows^{v, \phi}$, $\phi$-Conitzer, and
$\phi$-Walsh, for  $\phi\in \{0,0.001,\dots,1\}$.
We select the
distribution corresponding to the closest matrix. 

\begin{table}
\centering
 \caption{Closest distributions for seven illustrative real-world elections. For each, we include the parameterization that produces the closest frequency matrix, and give the normalized positionwise distance of the elections from this matrix.  }
    \begin{tabular}{cccccccccccc}
    \toprule
    id & source &
    \multicolumn{2}{c}{$\calD_\Mallows^{v, \phi}$} & \phantom{a} &
    \multicolumn{2}{c}{$\phi$-Walsh} & \phantom{a} & \multicolumn{2}{c}{$\phi$-Conitzer}\\ \cmidrule{3-4}
    \cmidrule{6-7} \cmidrule{9-10}
    & &$\phi$ & dist & & $\phi$ &dist & & $\phi$ &dist \\ \midrule
    1 & f. skate & 0.44 & 0.32 & & 0.14 & 0.38 & & 0.11 & 0.44 \\
    2 & f. skate & 0.23 & 0.24 & & 0 & 0.36 & & 0.02 & 0.56\\
    3 & aspen &  0.68 & 0.18 && 0.3 & 0.16 & & 0.23& 0.27  \\
    4 & f. skate & 0.05 & 0.09 && 0 & 0.37 & & 0 & 0.69  \\
    5 & s. skate &  0.46 & 0.11 && 0.15 & 0.18 & & 0.16 & 0.35  \\
    6 & irish &  0.75 & 0.1 && 0.42 & 0.12 & & 0.36 & 0.16  \\
    7 & cities &  0.93 & 0.06 && 0.69 & 0.06 & & 0.63 & 0.06  \\
    \bottomrule
\end{tabular}
    \label{tab:high}
\end{table}

\myparagraph{Data.}
We consider elections from the real-world datasets used by \citet{boe-bre-fal-nie-szu:t:compass}. They generated $15$ elections with $10$ candidates and $100$ voters (with strict preferences) from each of the eleven different real-world election datasets (so, altogether, they generated 165 elections, most of them from Preflib \citep{mat-wal:c:Preflib}). They used four datasets of political elections (from North Dublin (Irish), various non-profit and professional organizations (ERS), and city council elections from Glasgow and Aspen),
four datasets of sport-based elections (from Tour de France (TDF), Giro d’Italia
(GDI), speed skating, and figure skating) and three datasets with survey-based elections (from preferences over T-shirt designs, sushi, and cities).
We present the results of our analysis for seven illustrative and particularly interesting elections in \Cref{tab:high} and also include them in our skeleton map from \Cref{fig:skeleton_map}.

\myparagraph{Basic Test.}
There is a standard maximum-likelihood estimator
(MLE; based on Kemeny voting~\citep{MM09})
that given an election provides the most
likely dispersion parameter of the Mallows distribution that might have generated this election. 
To test our approach, we compared the parameters provided by
our approach and by the MLE for our $165$ elections and found them
to be highly
correlated (with Pearson correlation coefficient around $0.97$).
In particular, the average absolute difference between the dispersion parameter calculated by our approach and the MLE is only $0.02$.
See \Cref{se:validation} for details.

\myparagraph{Fitting Real-World Elections.}
Next, we consider the capabilities of $\calD_\Mallows^{v, \phi}$, $\phi$-Conitzer, and
$\phi$-Walsh to fit the real-world elections of \citet{boe-bre-fal-nie-szu:t:compass}. 
Overall, we find that these three vote distributions have some ability to
capture the considered elections,
but it certainly is not perfect.
Indeed, the average normalized distance of these elections to the frequency matrix of the closest distribution is $0.14$. To illustrate that some distance is to be expected here, we mention that the average distance of an election
sampled from impartial culture ($\calD_\IC$, with $10$ candidates
and $100$ voters) to the distribution's expected frequency matrix
is $0.09$ (see \Cref{app:variance} for a discussion of this and how it may serve as an estimator for the ``variance of a distribution'').
There are also some elections that are not captured by any of the considered distributions to an acceptable degree; examples of this are elections nr.\ 1 and nr.\ 2, which are at distance at least $0.32$ and $0.25$ from all our distributions, respectively.
Remarkably, while coming from the same dataset, elections nr.\ 1 and nr.\ 2 are still quite different from each other and, accordingly, the computed dispersion parameter is also quite different. 
It remains a challenge to find distributions capturing such elections. 

Comparing the power of the three considered models, nearly all of our elections are best captured by the Mallows model rather than $\phi$-Conitzer or $\phi$-Walsh.
There are only twenty elections that are closer to $\phi$-Walsh or $\phi$-Conitzer than to a Mallows model (election nr.\ 3 is the most extreme example), and, unsurprisingly, both $\phi$-Walsh and $\phi$-Conitzer perform particularly badly at capturing elections close to ID (see election nr.\ 4). 
That is, $\phi$-Conitzer and $\phi$-Walsh are not needed to ensure good coverage of the space of elections; the average normalized distance of our elections to the closest Mallows model is only $0.0007$ higher than their distance to the closest distribution (elections nr.\ 3-6 are three examples of elections which are well captured by the Mallows model
and distributed over the entire map).\footnote{For each election, we also computed the closest frequency matrix of two mixtures of Mallows models with reversed central votes $p$-$\calD_\Mallows^{v, \phi, \psi}$ using our approach. However, this only decreased the average minimum distance by around $0.02$, with the probability $p$ of flipping the central vote being (close to) zero for most elections.  }
Nevertheless, $\phi$-Walsh is also surprisingly powerful, as the average normalized distance of our elections to the closest $\phi$-Walsh distribution is only $0.03$ higher than their distance to the closest distribution (however, this might be also due to the fact that most of the considered real-world elections fall into the same area of the map, which $\phi$-Walsh happens to capture particularly well \citep{boe-bre-fal-nie-szu:t:compass}). $\phi$-Conitzer  performs considerably worse: there are only three elections for which it produces a (slightly) better result than $\phi$-Walsh. 

Moreover, our results also emphasize the complex nature of the space of elections:
Election nr.\ 7 is very close to $\calD_\Mallows^{v, 0.95}$, hinting that its
votes are quite chaotic. At the same time, this election is very close to $0.63$-Conitzer and $0.69$-Walsh distributions, which 
suggests at least a certain level of structure among its
votes (because votes from Conitzer and Walsh distributions are
very structured, and the Mallows filter with dispersion between
$0.63$ and $0.69$ does not destroy this structure fully).
However, as witnessed by the fact that the frequency matrix of
GS/balanced (which is highly structured) is $\UN$, such phenomena can happen.
Lastly, note that most of our datasets are quite ``homogenous'', in that the closest distributions for elections from the dataset are similar and also at a similar distance. 
However, there are also clear exceptions, for instance, elections nr.\ 1 and nr.\ 4 from the figure skating dataset. Moreover, there are two elections from the speed skating dataset where one election is captured best by $\calD_\Mallows^{v, 0.76}$ and the other by $\calD_\Mallows^{v, 0.32}$.

\section{Summary}
We have computed the frequency matrices \citep{szu-fal-sko-sli-tal:c:map,boe-bre-fal-nie-szu:t:compass} of several well-known distributions of votes.
Using them, we have drawn a ``skeleton map'', which shows how these
distributions relate to each other, and we have analyzed its properties.
Moreover, we have demonstrated how our results can be used to fit vote distributions to capture real-world elections. 

For future work, it would be interesting to compute the frequency matrices of further popular vote distributions, such as the Plackett--Luce model (we conjecture that its frequency matrix is computable in polynomial time). 
It would also be interesting to use our approach to fit more complex models, such as mixtures of Mallows models, to real-world elections. Further, it may be interesting to use expected frequency matrices to  reason about the asymptotic behavior of our models. For example, it might be possible to formally show where, in the limit, do the matrices of our models end up on the map as we increase the number of candidates.

\subsection*{Acknowledgments}
NB was supported by the DFG project MaMu (NI 369/19) and by the DFG project ComSoc-MPMS (NI 369/22). 
This project has received funding from the European 
    Research Council (ERC) under the European Union’s Horizon 2020 
    research and innovation programme (grant agreement No 101002854).

\begin{center}
  \includegraphics[width=3cm]{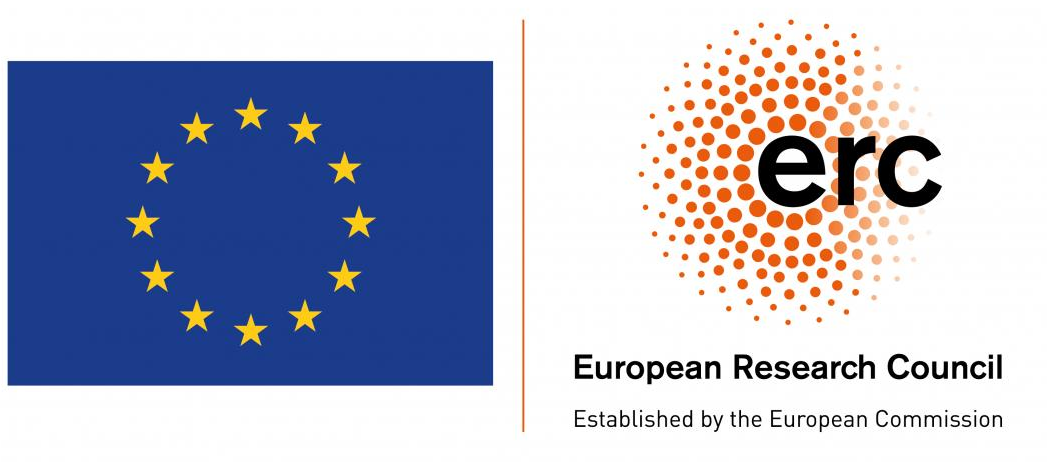}
\end{center}

\bibliographystyle{icml2022}

\newpage
\appendix

\section{Missing Material From \Cref{sec:sp-mat}}
Recall that our candidate set is $C = \{c_1, \ldots, c_m\}$ and the
societal axis is $c_1 \lhd \cdots \lhd c_m$.
We consider the Conitzer distribution.
Let $f(i,j)$ be the
probability that in a sampled vote candidates
$c_i, \ldots, c_j$ appear in the top $j-i+1$ positions. Next we find the values 
of $f(i, j)$ for all $i, j\in [m]$, and using them
we establish the frequency matrix for the Conitzer distribution.

\begin{restatable}{proposition}{contop}\label{prop:con:top}
Let $i, j$ be two integers with $1 < i \leq j < m$.  Then
  $f(\ell,\ell) = \nicefrac{1}{m}$ for all $\ell\in[m]$, 
  $f(1,m) = 1$,
  $f(i,j) = \nicefrac{1}{m}$,
  $f(1,i) = \nicefrac{(i+1)}{2m}$, and
  $f(j,m) = \nicefrac{(m-j+2)}{2m}$.
\end{restatable}

\begin{proof}
  The quantity $f(\ell, \ell)$ is simply the probability that $c_\ell$ is ranked first, so we have $f(\ell,\ell) = \nicefrac{1}{m}$ by the definition of the Conitzer distribution.
  
  The equality $f(1, m)=1$ is immediate from the definition of $f(i, j)$.

  To show that $f(i,j) = \nicefrac{1}{m}$, we give a proof by induction
  on $j-i$.  The base step holds because for each $\ell \in [m]$ we have
  $f(\ell,\ell) = \nicefrac{1}{m}$. Assume that for all integers $x$, $y$
  such that $1 < x \leq y < m$ and $y-x < j-i$ we have
  $f(x,y) = \nicefrac{1}{m}$. The only way for candidates
  $c_i, \ldots, c_j$ to be ranked in top $j-i+1$ positions under the Conitzer model is
  that, while generating the vote, we placed candidates
  $c_i, \ldots, c_{j-1}$ in top $j-i$ positions and then extended
  the vote with $c_j$ (the probability of this latter step
  is $\nicefrac{1}{2}$), or we placed candidates $c_{i+1}, \ldots, c_{j}$ in top $j-i$ positions and then extended the vote with $c_i$ (again, the probability of the latter step is $\nicefrac{1}{2}$). Thus:
  \[ \textstyle
    f(i,j) = \frac{1}{2} f(i,j-1) + \frac{1}{2}f(i+1,j) =
    \frac{1}{2m} + \frac{1}{2m} = \frac{1}{m}.
  \]
  This proves the claim for $f(i,j)$.

  Next we show that $f(1,i) = \nicefrac{(i+1)}{2m}$, using induction
  over~$i< m$. For $i=1$ we have
  $f(1,1) = \nicefrac{1}{m} = \nicefrac{2}{2m}$, so our claim
  holds. Assume that it holds up to $i-1$.
  There are only two ways in which candidates $c_1, \ldots, c_i$ can
  be ranked in the top $i$ positions: Either we first place
  $c_1, \ldots c_{i-1}$ in top $i-1$ positions
  and then extend the vote with $c_i$ (the latter step has probability~$1$, since
  there is no candidate to the left of $c_1$), or we first place candidates $c_2, \ldots, c_{i}$ in the top $i-1$ positions
  and then extend the vote with $c_1$ (the latter step has probability
  $\nicefrac{1}{2}$, since the vote could also be extended with
  $c_{i+1}$). Thus, we have:
  \[ \textstyle
    f(1,i) = f(1,i-1) + \frac{1}{2}f(2,i) = \frac{i}{2m} + \frac{1}{2m} = \frac{i+1}{2m}.
  \]
  This completes the proof for $f(1,i)$. The expression for $f(j,m)$ can be derived
  by symmetry: we have $f(j, m)=f(1, m-j+1)$.
\end{proof}

\conmain*
\begin{proof}
  Consider a candidate $c_j$, $j \in [\frac{m}{2}]$, and 
  a position $i \in [m]$. We proceed by case analysis:
  \begin{enumerate}
  \item If $i < j$, then there are two ways to generate a vote with
    $c_j$ in position $i$: Either candidates
    $c_{j+1}, \ldots, c_{i+j-1}$ are ranked in the top $i-1$ positions
    or candidates $c_{j-i+1}, \ldots, c_{j-1}$ are ranked in the top
    $i-1$ positions. In both cases, $c_j$ is ranked in the $i$-th
    position with probability $\nicefrac{1}{2}$ (indeed, we have $i+j <\nicefrac{m}{2}+\nicefrac{m}{2} = m$, so in the former case both $c_j$ and $c_{i+j}$ could have been placed in position $i$, and also $j-i\ge 1$, so in the latter case both $c_j$ and $c_{j-i}$ could have been placed in position $i$). Thus we have:
    \begin{align*}
      \textstyle
      \prob[ \pos_v&(c_j) = i ] =
      \textstyle\frac{1}{2}\cdot f(j+1,i+j-1)\\&+
      \textstyle\frac{1}{2}\cdot f(j-i+1,j-1) = \frac{1}{2m} +
      \frac{1}{2m} = \frac{1}{m}.
    \end{align*}
  \item If $i = j$, then either candidates $c_1, \ldots, c_{j-1}$ are
    ranked in the top $j-1$ positions and the vote is extended with
    $c_j$ (with probability $1$), or candidates
    $c_{j+1}, \ldots, c_{2j-1}$ are ranked in the top $j-1$ positions
    and the vote is extended with $c_j$ (with probability
    $\nicefrac{1}{2}$). Thus, we have:
    \begin{align*}
      \prob[ \pos_v(c_j) = j ] &=  f(1,j-1) +
      \textstyle\frac{1}{2}\cdot f(j+1,2j-1) \\  &= \textstyle\frac{j}{2m} +
      \frac{1}{2m} = \frac{j+1}{2m}.
    \end{align*}
  \item If $j < i < m-j+1$ then there is only one possibility
    for $c_j$ to be ranked $i$-th: It must be that candidates
    $c_{j+1}, \ldots, c_{j+i-1}$ are ranked in the top $i-1$ positions
    and the vote is extended with $c_j$ (which happens with
    probability $\frac{1}{2}$ because $j+i \le m$).\footnote{It is
      impossible for the candidates from the left side of $c_j$ to
      take the top $i-1$ positions because there are fewer than $i-1$ of
      them.}  Thus, we have:
    \[
      \textstyle
      \prob[ \pos_v(c_j) = i ] =  \frac{1}{2}\cdot f(j+1,i+j-1) = \frac{1}{2m}.
    \]
  \item If $i = m-j+1$, the analysis is similar to the case $i=j$:
    candidates $c_{j+1}, \ldots, c_m$ must be ranked in the top $m-j$ positions,
    in which case $c_j$ gets ranked in the $(m-j+1)$-st position (with
    probability~$1$). Thus, we have:
    \[
      \prob[ \pos_v(c_j) = m-j+1 ] =  f(j+1,m) = \frac{m-j+1}{2m}.
    \]
  \item If $i > m-j+1$ then $\prob[ \pos_v(c_j) = i ] = 0$, because
    both to the left of $c_j$ and to the right of $c_j$ there are fewer
    than $i-1$ candidates. 
  \end{enumerate}
  The fact that for each candidate $c_j \in C$ and each position
  $i \in [m]$ we have
  $\prob[ \pos_v(c_j) = i] = \prob[ \pos_v(c_{m-j+1}) = i]$ follows
  directly from the symmetry of the Conitzer distribution and the fact
  that $m$ is even.
\end{proof}

\section{Missing Material From \Cref{sec:Mallows}}\label{app:mallows}
\begin{proposition}[\citet{oeis}]
  There is an algorithm that computes $S(m,k)$ using at most
  polynomially many operations.
\end{proposition}
\begin{proof}
  First, we note that for each $m' \in [m]$ we have $S(m',0) =
  1$. Further, for each $m' \in [m]$ and $k' \in [k]_0$ the following
  recursion holds:
  \[ S(m',k')=S(m',k'-1)+S(m'-1,k')-S(m'-1,k'-m').\] Using these two
  facts and standard dynamic programming, we can compute $S(m,k)$
  using $O(mk)$ arithmetic operations. Since $k$ is at most $O(m^2)$,
  the running time is at most~$O(m^3)$.
\end{proof}

\begin{figure}[!htb]
\centering
 	\begin{minipage}{.4\textwidth}
 	    \resizebox{\textwidth}{!}{\input{compass_dist_ConWal}}   \caption{Normalized positionwise distance between the Conitzer
    [Walsh] matrix  and the compass
    matrices in solid [dashed] lines, for varying number of
    candidates.}\label{fig:distancesCompass}
 	\end{minipage}\hfill
    \centering
    \begin{minipage}{.4\textwidth} 
        \centering
        \resizebox{\textwidth}{!}{\input{compass_dist_caterpillar}} 		\caption{Normalized positionwise distance between the frequency matrix of the GS/caterpillar distribution and the compass matrices, for varying number of candidates.}\label{fig:compass_distance_GScater}
    \end{minipage}\\
    \begin{minipage}{0.4\textwidth}
        \centering
        \resizebox{\textwidth}{!}{\input{compass_dist_Mallows}}  		\caption{For different values of $\phi$, normalized positionwise distance between the frequency matrix of the Mallows distribution and the compass matrices, for varying number of candidates.}\label{fig:compass_distance_Mallows}
    \end{minipage}\hfill
    \begin{minipage}{0.4\textwidth}
        \centering
        \resizebox{\textwidth}{!}{\input{compass_dist_normMallows}}  		\caption{For different values of $\normphi$, normalized positionwise distance between the frequency matrix of the normalized Mallows distribution and the compass matrices, for varying number of candidates.}\label{fig:compass_distance_normalizedMallows}
    \end{minipage}
\end{figure}

\section{Missing Material From \Cref{sec:robustness}}\label{app:robustness}
\subsection{Distance to Compass Matrices} \label{sub:distancesCompass}
We analyze the distances between our matrices
for different numbers of candidates.
In \Cref{fig:distancesCompass} we show these distances for the
Conitzer and Walsh matrices and the compass matrices: For Conitzer, they are nearly constant,
and for Walsh they vary significantly. Indeed, the more
candidates we have, the closer the Walsh matrix is to $\ID$ (e.g., for
$10$ candidates their distance is $0.44$, and for $300$ candidates it
is $0.09$).

\Cref{fig:compass_distance_GScater} depicts the distance between the frequency matrix for GS/caterpillar and the four compass matrices, for varying number of candidates. As for the matrix for the Walsh model, its distance to the compass matrices changes as the number of candidates increases: The matrix moves closer and closer to $\AN$.

In \Cref{fig:compass_distance_Mallows}, we display the distance between the frequency matrix for the Mallows model for different values of the dispersion parameter $\phi$ and the compass matrices (in contrast to the previous figures, we only consider up to $100$ candidates, as for more than $100$ candidates computing the matrix for the Mallows model becomes very memory-consuming). Independent of the chosen value of the dispersion parameter, the distance of the respective matrix to the four compass matrices changes significantly when we increase the number of candidates. In fact, for any fixed dispersion parameter $\phi$, the resulting matrix will always move closer and closer to $\ID$ as the number of candidates increases. 

In contrast, if we use the normalized version of the Mallows model, the matrices remain more or less at a constant distance from the compass matrices. \Cref{fig:compass_distance_normalizedMallows} shows the distance of the frequency matrix of the normalized version of the Mallows model for different values of $\normphi$, as the number of candidates increases. 

\subsection{Distances of Pairs of Vote Distributions on the Skeleton Map} \label{sec:maxDiff}
\begin{figure}
\centering
  \resizebox{0.33\textwidth}{!}{
\begin{tikzpicture}

\definecolor{color0}{rgb}{0.12156862745098,0.466666666666667,0.705882352941177}

\begin{axis}[
tick align=outside,
tick pos=left,
x grid style={white!69.0196078431373!black},
xlabel={number of candidates},
xmin=-0.800000000000001, xmax=104.8,
xtick style={color=black},
y grid style={white!69.0196078431373!black},
ylabel={maximum pairwise difference},
ymin=-0.0142568626270305, ymax=0.29939411516764,
ytick style={color=black},
ytick={-0.05,0,0.05,0.1,0.15,0.2,0.25,0.3},
yticklabels={−0.05,0.00,0.05,0.10,0.15,0.20,0.25,0.30}
]
\addplot [line width=1.5 , color0]
table {%
4 0.112986564660799
6 0.0780361679100244
8 0.0652745568728313
10 0.0547865216590647
12 0.0469007594574643
14 0.0407035695947612
16 0.0358444925738024
18 0.0319005221134383
20 0.0287234278998958
22 0.0259555009520156
24 0.0235551691045071
26 0.0215336088908751
28 0.019711155231833
30 0.0181156851810114
32 0.0166701496115702
34 0.0153563871417951
36 0.0141629655345483
38 0.0130988453679295
40 0.0121056284803952
42 0.011217798487155
44 0.0104090896914342
46 0.00964495020215217
48 0.00894178217257968
50 0.00828739010062018
52 0.00769651628588094
54 0.00714894748176179
56 0.00663375260298824
58 0.00614825077646719
60 0.00568998375520524
62 0.00525677583567807
64 0.00484665230389036
66 0.00445785195610959
68 0.00408877289675408
70 0.00373798794277827
72 0.00340415495705781
74 0.00308608692613543
76 0.00278271451276946
78 0.0024930721445901
80 0.00221622020253265
82 0.00195137065228909
84 0.00169772635391088
86 0.00145464056619574
88 0.00122143060345942
90 0.000997537018820438
92 0.000782415800247183
94 0.000575534982853648
96 0.000376462890199464
98 0.000184753524552467
100 0
};
\addlegendentry{absolute difference}
\addplot [line width=1.5 , red]
table {%
4 0.285137252540609
6 0.253341217343897
8 0.201296875581962
10 0.169653967009331
12 0.149197425967739
14 0.141302558564751
16 0.132353828977743
18 0.123250784592509
20 0.114435448180624
22 0.106112345701684
24 0.0983614957084886
26 0.0911951691094657
28 0.0845917539364811
30 0.0785154488057373
32 0.0729230698167747
34 0.0677710559109892
36 0.0630178104823855
38 0.0586242956910079
40 0.0545550962403525
42 0.0507786811963612
44 0.0472662665326878
46 0.0439930084283992
48 0.0409362334428935
50 0.0380759637849896
52 0.0353943928314437
54 0.0328762580500243
56 0.0305070030199907
58 0.028274299062774
60 0.0261668413026967
62 0.0241746242122755
64 0.0222885665656306
66 0.0205005690182728
68 0.0188032648448699
70 0.0171900907802407
72 0.01565487466456
74 0.0141921577137438
76 0.0127970223077944
78 0.0114650280159385
80 0.0101918537602943
82 0.00897387556413975
84 0.00780742757614246
86 0.00668953559195808
88 0.00561706010737579
90 0.00458742836324637
92 0.00359813858151375
94 0.00264674438599315
96 0.00173126060249158
98 0.000849636196703676
100 0
};
\addlegendentry{relative difference}
\end{axis}
\end{tikzpicture}} \caption{Results from our experiments described in \Cref{sec:maxDiff}} \label{fig:maxDiff}
\end{figure}
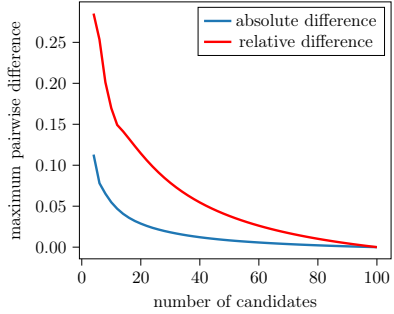

\begin{figure*}[h!]
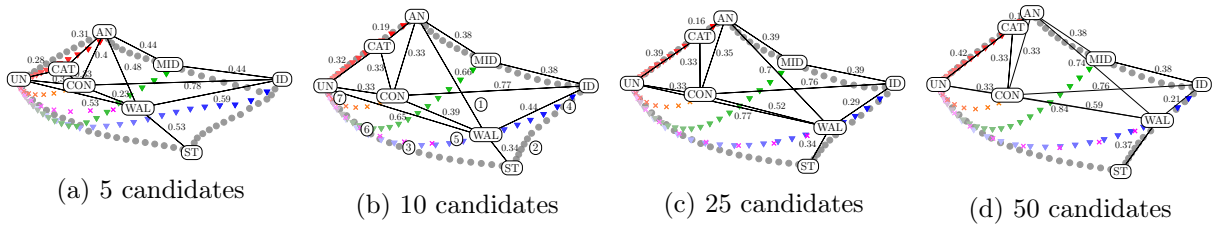

\centering
\begin{subfigure}{.24\textwidth}
  \centering
  \resizebox{\textwidth}{!}{\input{skeleton_5}}   \caption{$5$ candidates}
  \label{fig:m5}
\end{subfigure}\hfill
\begin{subfigure}{.24\textwidth}
  \centering
  \resizebox{\textwidth}{!}{\input{skeleton_10}}   \caption{$10$ candidates}
  \label{fig:m10}
\end{subfigure}\hfill
\begin{subfigure}{.24\textwidth}
  \centering
  \resizebox{\textwidth}{!}{\input{skeleton_25}}   \caption{$25$ candidates}
  \label{fig:m20}
\end{subfigure}\hfill
\begin{subfigure}{.24\textwidth}
  \centering
  \resizebox{\textwidth}{!}{\input{skeleton_50}}   \caption{$50$ candidates}
  \label{fig:m50}
\end{subfigure}
\caption{Skeleton map for different number of candidates.}
\label{fig:skeletonMapVaryingM}
\end{figure*}

As we have observed above, some vote distributions stay more or less at constant normalized positionwise distance from the four compass matrices. This raises the question whether these matrices also stay at a constant normalized positionwise distance from each other. This would imply that the data on which the skeleton map is based is independent of the number of candidates, and thus that the map is likely to look very similar for different numbers of candidates. 
To check this, we conducted the following experiment. We put together a set of vote distributions/matrices that do not structurally change when increasing the number of candidates (like the change happening for the Walsh model). First, we add the four compass matrices and the frequency matrix of the Conitzer model. Second, we add the frequency matrices of different variants of the Mallows model (similar as on the skeleton map as described in \Cref{sec:experiments}): the normalized Mallows model, the normalized Mallows model where the central vote is reversed with probability $\nicefrac{1}{2}$, the normalized Mallows model where the central vote is reversed with probability $\nicefrac{1}{4}$, and the distribution where we first sample a vote $v$ from the Conitzer distribution and then sample the final vote from the normalized Mallows model with $v$ as the central vote. For each of these variants, we include their frequency matrix for $\normphi\in \{0.2,0.4,0.6,0.8\}$.

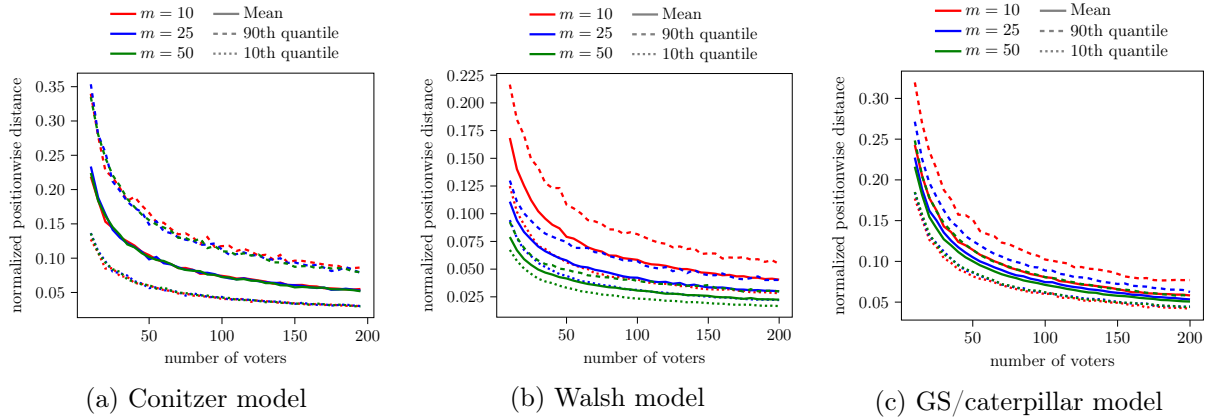
\begin{figure*}
\centering
\begin{subfigure}{.32\textwidth}
  \centering
  \resizebox{\textwidth}{!}{
\begin{tikzpicture}

\begin{axis}[
legend columns=2, 
legend cell align={left},
legend style={
	fill opacity=0.8,
	draw opacity=1,
	draw=none,
	text opacity=1,
	at={(0.5,1.3)},
	line width=1.5pt,
	anchor=north,
	/tikz/column 2/.style={
		column sep=10pt,
	}
},
legend entries={$m=10$,
	Mean,
	$m=25$,
	90th quantile,
	$m=50$,
	10th quantile},
tick align=outside,
tick pos=left,
x grid style={white!69.0196078431373!black},
xlabel={number of voters},
xmin=0.75, xmax=204.25,
xtick style={color=black},
y grid style={white!69.0196078431373!black},
ylabel={normalized positionwise distance},
ymin=0.0133083625567369, ymax=0.369739916072,
ytick style={color=black},
ytick={0,0.05,0.1,0.15,0.2,0.25,0.3,0.35,0.4},
yticklabels={0.00,0.05,0.10,0.15,0.20,0.25,0.30,0.35,0.40}
]
\addlegendimage{red}
\addlegendimage{gray}
\addlegendimage{blue}
\addlegendimage{gray,dashed}
\addlegendimage{green!50.1960784313725!black}
\addlegendimage{gray,dotted}
\addplot [line width=1.5, red]
table {%
10 0.218747493662316
15 0.184523595226458
20 0.153095975970341
25 0.144091937938009
30 0.130769377003678
35 0.124462497382509
40 0.118308095005482
45 0.109400127022386
50 0.104440420597528
55 0.098128114162916
60 0.093897317642854
65 0.0913846272406065
70 0.0855267085211182
75 0.0857070836045713
80 0.0830404173070095
85 0.0772454061922664
90 0.0775376088950181
95 0.0753740145883175
100 0.0742010236865455
105 0.0727188652326739
110 0.0702943208475945
115 0.070683372841226
120 0.0691894051348885
125 0.0674889007498595
130 0.0658473285149331
135 0.0627510414610269
140 0.0644134292577516
145 0.0608631249724163
150 0.0607609549448872
155 0.0609739454311667
160 0.0592020332475541
165 0.0587358557279551
170 0.0561604362412008
175 0.0559710075606526
180 0.0553995595751989
185 0.0545301792863756
190 0.0543293569844705
195 0.0550058355103388
};
\addplot [line width=1.5, red, dashed]
table {%
10 0.339696988700466
15 0.277979819250829
20 0.227575784345919
25 0.217636387592012
30 0.202020222942034
35 0.183593089263322
40 0.189393948582989
45 0.17676768638758
50 0.166121233604623
55 0.152672199363058
60 0.15161617234575
65 0.142657363900181
70 0.13476192054875
75 0.136202040088899
80 0.130378806252371
85 0.120623903027312
90 0.130740761587566
95 0.116427451902718
100 0.117575775233634
105 0.118037526958594
110 0.110495882797422
115 0.113610035801927
120 0.109595977481116
125 0.103806071784912
130 0.104941733136322
135 0.0983613987086397
140 0.104805209137725
145 0.0939707532299287
150 0.0963636463112903
155 0.0960508505221118
160 0.0943560701766701
165 0.092617097982403
170 0.0907486714761366
175 0.0894026141444391
180 0.088585874009313
185 0.0843652917580171
190 0.0854864509832678
195 0.0863869535990737
};
\addplot [line width=1.5, red, dotted]
table {%
10 0.127272742568995
15 0.107979815949996
20 0.0848485045022134
25 0.0853939470355258
30 0.0736363732453549
35 0.0727272910884384
40 0.0636363831546271
45 0.0625589323201866
50 0.0599394047463482
55 0.0567493265212485
60 0.0535353618137764
65 0.0522144636873043
70 0.0484848665073514
75 0.0486666759814728
80 0.0453788015370568
85 0.045258477510828
90 0.0434343558877255
95 0.0427432370061676
100 0.0406060773032633
105 0.0402597518016895
110 0.0399449169635773
115 0.0405533779976946
120 0.0363636514849284
125 0.0387757676230236
130 0.0372727324508808
135 0.036352425695143
140 0.0345887539005189
145 0.0354232018982822
150 0.0341212221788186
155 0.0333137868310917
160 0.032196998099486
165 0.0333149816163562
170 0.031711239644298
175 0.0322857255849874
180 0.0313131377608939
185 0.0305405588199695
190 0.0307655533041918
195 0.0299844685659716
};
\addplot [line width=1.5, blue]
table {%
10 0.233603562030237
15 0.188320457205857
20 0.16409682206923
25 0.141377570705417
30 0.129931558455096
35 0.121626403554632
40 0.116067329180302
45 0.109293231323409
50 0.0983019303999101
55 0.101777020128107
60 0.0926224643404846
65 0.0908285513310619
70 0.0867577656554595
75 0.0831035451018639
80 0.0828480142401158
85 0.078592186378127
90 0.0787282727281518
95 0.0766586517103506
100 0.0728897533480538
105 0.0706092477251047
110 0.0692141705908057
115 0.0705946023807001
120 0.0681654312930933
125 0.0663816199236466
130 0.0648443549149894
135 0.0641458213097156
140 0.0635813871231855
145 0.0610397760855416
150 0.0604793958086297
155 0.0595994827421442
160 0.057600793741421
165 0.0565144729517842
170 0.0571712392174045
175 0.0539157609183204
180 0.0542711419258894
185 0.0554741103140912
190 0.0541147391569654
195 0.0527736461483829
};
\addplot [line width=1.5, blue, dashed]
table {%
10 0.353538481821306
15 0.28008337789036
20 0.253971185425941
25 0.211749998026957
30 0.198121840189784
35 0.182574210478029
40 0.172966366837947
45 0.167305592657067
50 0.149038467106696
55 0.153730780369817
60 0.139355797608061
65 0.136368364937908
70 0.130598937343502
75 0.126897463877685
80 0.12528127029968
85 0.11945252285133
90 0.119897472879921
95 0.121499001378946
100 0.112538465750941
105 0.105018344197351
110 0.105982528688261
115 0.108448164499938
120 0.104379829167066
125 0.104011560579798
130 0.0989571343502925
135 0.0977906108775642
140 0.0999011220902992
145 0.0930928584725524
150 0.0937307865358889
155 0.090980161614537
160 0.0887343873032548
165 0.090378804199505
170 0.0862149550512326
175 0.0812610048044007
180 0.0828023795880234
185 0.0868867047507173
190 0.0831285510601727
195 0.0783284389664634
};
\addplot [line width=1.5, blue, dotted]
table {%
10 0.136480807380464
15 0.110448765804848
20 0.0953654030588671
25 0.0803846209309995
30 0.0781090147273902
35 0.0746538708676011
40 0.0663894477649592
45 0.0630363564875622
50 0.0563461635691615
55 0.0576014063947906
60 0.0532019449854628
65 0.0529023856699216
70 0.0506868414956933
75 0.0454359136683007
80 0.048841360435248
85 0.046424228095342
90 0.0438867807571884
95 0.0428532576134715
100 0.0435192381622843
105 0.0393663241998113
110 0.0404108468711806
115 0.0398352921877701
120 0.0390945709664196
125 0.037688479174036
130 0.0370636336062139
135 0.036745027165251
140 0.0365838125867483
145 0.0348116803179997
150 0.0333782215275838
155 0.0348691161614485
160 0.0333353457054517
165 0.0319784557509523
170 0.0329762709022571
175 0.0308681417314801
180 0.0317553646750569
185 0.0322889912282475
190 0.0309534480522136
195 0.0302091003274724
};
\addplot [line width=1.5, green!50.1960784313725!black]
table {%
10 0.224175588900623
15 0.183871235587172
20 0.161873307428388
25 0.145274167540083
30 0.133970749913299
35 0.120089428776889
40 0.115734685204346
45 0.107186513670056
50 0.101761798447728
55 0.0974137363755847
60 0.0959589598721349
65 0.0914062667339884
70 0.0856861444807709
75 0.0828808792208393
80 0.0815996582811216
85 0.0810598877258729
90 0.0761311493773183
95 0.0757016058314313
100 0.0724134955715601
105 0.0718342493738712
110 0.0688439327606666
115 0.0695481164016019
120 0.0671603548863031
125 0.0668642999535101
130 0.0636313725108256
135 0.0635646547341109
140 0.0605530485983047
145 0.0601819435755784
150 0.0574066802053502
155 0.0580682198786061
160 0.0574071353236681
165 0.0568557743241369
170 0.0569900331123905
175 0.0547013249483053
180 0.0544775062664386
185 0.0542640139349625
190 0.0534121266008802
195 0.0517956902865576
};
\addplot [line width=1.5, green!50.1960784313725!black, dashed]
table {%
10 0.33479236018333
15 0.280391395551299
20 0.246182505420375
25 0.220801941139483
30 0.205197323893397
35 0.186473376882159
40 0.174528237833839
45 0.163518251155113
50 0.15544299625066
55 0.14902064088969
60 0.144768358987956
65 0.139893101925748
70 0.129734566429628
75 0.129483048186949
80 0.124290554166748
85 0.128549170542826
90 0.114688849538265
95 0.117617013340085
100 0.109973600339077
105 0.10949332298541
110 0.104903214359108
115 0.107042043768809
120 0.103833970132017
125 0.105414689762019
130 0.0986477378858108
135 0.0938786364252157
140 0.0937846310250014
145 0.0917722608303788
150 0.0873781975047911
155 0.0882852774619327
160 0.0885372452150543
165 0.085095509167622
170 0.087341615649573
175 0.0849202700601589
180 0.0854918346934284
185 0.0844349886584772
190 0.0806826520742489
195 0.0796807055839529
};
\addplot [line width=1.5, green!50.1960784313725!black, dotted]
table {%
10 0.13605766359578
15 0.107963253628565
20 0.0936470894952442
25 0.0838751663926806
30 0.0800848859553488
35 0.069089043004682
40 0.0651128751258193
45 0.0616937781103841
50 0.06077312589528
55 0.0554990856837826
60 0.0574238024658712
65 0.0526390515064367
70 0.04978908483422
75 0.0491444758860683
80 0.0470798517191628
85 0.0455626324569361
90 0.045551050735452
95 0.0429701375284222
100 0.0423337392768368
105 0.0419946565345323
110 0.0392290843059259
115 0.0389110189810002
120 0.0390664563749079
125 0.037793539793064
130 0.0357650956755005
135 0.0368721800367711
140 0.035043750242537
145 0.0346187937767709
150 0.0334974139033272
155 0.0340287425768274
160 0.0321230627331963
165 0.0338974962883046
170 0.0327717285158866
175 0.0310262501871466
180 0.0321819705007618
185 0.0315320180208028
190 0.0301061653234691
195 0.0295097968074306
};
\end{axis}

\end{tikzpicture}}   \caption{Conitzer model}
  \label{fig:CitySizeCon}
\end{subfigure}\hfill
\begin{subfigure}{.32\textwidth}
  \centering
  \resizebox{\textwidth}{!}{
\begin{tikzpicture}

\begin{axis}[
legend columns=2, 
legend cell align={left},
legend style={
	fill opacity=0.8,
	draw opacity=1,
	draw=none,
	text opacity=1,
	at={(0.5,1.3)},
	line width=1.5pt,
	anchor=north,
	/tikz/column 2/.style={
		column sep=10pt,
	}
},
legend entries={$m=10$,
	Mean,
	$m=25$,
	90th quantile,
	$m=50$,
	10th quantile},
tick align=outside,
tick pos=left,
x grid style={white!69.0196078431373!black},
xlabel={number of voters},
xmin=0.75, xmax=209.25,
xtick style={color=black},
y grid style={white!69.0196078431373!black},
ylabel={normalized positionwise distance},
ymin=0.00664459665932848, ymax=0.226489166962571,
ytick style={color=black},
ytick={0,0.025,0.05,0.075,0.1,0.125,0.15,0.175,0.2,0.225,0.25},
yticklabels={0.000,0.025,0.050,0.075,0.100,0.125,0.150,0.175,0.200,0.225,0.250}
]
\addlegendimage{red}
\addlegendimage{gray}
\addlegendimage{blue}
\addlegendimage{gray,dashed}
\addlegendimage{green!50.1960784313725!black}
\addlegendimage{gray,dotted}
\addplot [line width=1.5, red]
table {%
10 0.168326400286622
15 0.140119986297055
20 0.125351927520877
25 0.112237768199182
30 0.102257082841049
35 0.0955467348214653
40 0.0900425020385195
45 0.08679930816508
50 0.0791544803557447
55 0.0775700282670482
60 0.0738516558458408
65 0.0697275444733995
70 0.0673703342386418
75 0.0658255173145521
80 0.0623208311543508
85 0.0620889143706911
90 0.0599042371525006
95 0.0591664255734985
100 0.0583821080206432
105 0.0557075821307742
110 0.05433053700291
115 0.0539124376751539
120 0.0529718078293772
125 0.0513715944685644
130 0.0500676323209814
135 0.0500557995004805
140 0.0488186980674109
145 0.0469772429922313
150 0.0465592625134182
155 0.0458619856392504
160 0.0444736880406671
165 0.044798452375325
170 0.0434685520545551
175 0.0435555722548466
180 0.0422454068642768
185 0.0422297962543301
190 0.0410588049582403
195 0.0413805913866375
200 0.0405057659232048
};
\addplot [line width=1.5, red, dashed]
table {%
10 0.216496231948788
15 0.184908478142637
20 0.170840447554083
25 0.149277936870402
30 0.141668262502009
35 0.127742170282837
40 0.123196027021516
45 0.123113968584574
50 0.108191762390462
55 0.104876894989248
60 0.100499541656763
65 0.0942575507769079
70 0.0931026934420295
75 0.0906354247846387
80 0.0858877900762089
85 0.0856674907305701
90 0.0819610281536976
95 0.0832244386501385
100 0.0815161060372537
105 0.0788055874661288
110 0.0760124020899336
115 0.0739099631339989
120 0.0728574906504064
125 0.0696803079099592
130 0.0704616658375458
135 0.0681119402164989
140 0.0687635396762441
145 0.0636984906300451
150 0.0642193892716684
155 0.0640813657830498
160 0.0595714997570736
165 0.0605039802916122
170 0.0601841237565333
175 0.0607742855000789
180 0.0586192781853518
185 0.0589631381811518
190 0.0565735954763086
195 0.0579948691905222
200 0.0547708372980582
};
\addplot [line width=1.5, red, dotted]
table {%
10 0.124888742918318
15 0.101730606592063
20 0.0901420586488464
25 0.0800819166230433
30 0.0707339189156438
35 0.0681764274384036
40 0.0630681889470328
45 0.0602648995590932
50 0.0560667682439089
55 0.0539865289899436
60 0.0514038956639442
65 0.0500074811533771
70 0.0463014352332914
75 0.0453997838508451
80 0.0430161062204702
85 0.0429757341093412
90 0.0424468822579718
95 0.0410018022990588
100 0.0408039833390803
105 0.0385558954383614
110 0.037384219021734
115 0.0371079474461801
120 0.0367716261919475
125 0.0361923349456805
130 0.0346884280993519
135 0.0347308207754836
140 0.0332372992257164
145 0.0324079210672415
150 0.0318226066700211
155 0.031681641189833
160 0.0314394012093544
165 0.0318084300411018
170 0.0304956428692535
175 0.0305649418835387
180 0.0290967130779543
185 0.0291791489347816
190 0.0291497252994415
195 0.0286488310941918
200 0.0282750998397894
};
\addplot [line width=1.5, blue]
table {%
10 0.110877512001898
15 0.0939781714522113
20 0.0835211033293774
25 0.0768120907976603
30 0.070410797921284
35 0.0661501428961921
40 0.0629556546004889
45 0.0598712125957872
50 0.057840093079841
55 0.0541439557497431
60 0.0522662785559451
65 0.0515072409587073
70 0.0493812622041785
75 0.0482297714566513
80 0.0467694217519751
85 0.0450903906196711
90 0.0435019149167839
95 0.0424424277112718
100 0.0422681480188084
105 0.0405995054702567
110 0.0389287961652138
115 0.038250009428653
120 0.0385226061236179
125 0.0368329646608431
130 0.0375246977480539
135 0.0352015288925031
140 0.0359943435707529
145 0.0345534972548604
150 0.0335540603088493
155 0.0333838294281882
160 0.0334345434289478
165 0.0326673520158809
170 0.0316332907355629
175 0.031527614858944
180 0.0309250054235581
185 0.0310903265888098
190 0.0305174951409796
195 0.0304516899264276
200 0.0300282938729446
};
\addplot [line width=1.5, blue, dashed]
table {%
10 0.129915504469178
15 0.110884475661442
20 0.100396413055177
25 0.0936217593458983
30 0.0872563999038763
35 0.0820287821796508
40 0.0793015748101215
45 0.0765873231298218
50 0.074554438287249
55 0.0691852758253495
60 0.0690975253891128
65 0.0691095646103629
70 0.0653014461187503
75 0.0647616198313279
80 0.0617415708328526
85 0.0602290940269505
90 0.057982233479225
95 0.0577015048143669
100 0.0571822462110924
105 0.0532578423807326
110 0.0515353212930047
115 0.0506237802041981
120 0.0518355644214003
125 0.0481118603090111
130 0.0502640521583649
135 0.0464074261074599
140 0.0493717546711336
145 0.0457724167481781
150 0.0445488052563563
155 0.0444731909499611
160 0.0460634013208059
165 0.0435389639917188
170 0.0414235316225901
175 0.0408385522218081
180 0.0418608823157578
185 0.043414353852411
190 0.040178531978745
195 0.040218030102551
200 0.0402393863536417
};
\addplot [line width=1.5, blue, dotted]
table {%
10 0.0942340042370443
15 0.0779928991284508
20 0.0672905005210151
25 0.0617950482652164
30 0.0555983675966183
35 0.0519297415695082
40 0.0488157513273808
45 0.0455640672329957
50 0.0438888906250493
55 0.0415522299832306
60 0.0391846159965588
65 0.0384034955493497
70 0.036984999172497
75 0.0356440963071341
80 0.0349885828210972
85 0.0339505889264938
90 0.0324574941505069
95 0.0313163405660397
100 0.0313598035557124
105 0.030421494417197
110 0.0295978019976666
115 0.029056113236136
120 0.028797964970223
125 0.0277615506887042
130 0.0275672724087669
135 0.0264679246994022
140 0.0263663390879698
145 0.0255658810009034
150 0.0252591272578754
155 0.0245821999213569
160 0.0245426740685406
165 0.0240813989084787
170 0.0236008138648825
175 0.0230705133746736
180 0.0227453310870959
185 0.0224759393980583
190 0.0224373665051714
195 0.0223273637171727
200 0.0221042900673078
};
\addplot [line width=1.5, green!50.1960784313725!black]
table {%
10 0.0789585733298669
15 0.0670387287021703
20 0.0590204396219656
25 0.0538474062301495
30 0.0496515650929053
35 0.0468802557860582
40 0.0450560457189882
45 0.0427547764708781
50 0.041573312558623
55 0.039584726821469
60 0.0380205622788662
65 0.0366018885070128
70 0.035682914190552
75 0.034623848009061
80 0.0336263457196967
85 0.0332297528354469
90 0.032132464604022
95 0.0317073855404347
100 0.0306461931386015
105 0.0303268590364407
110 0.0297506332541827
115 0.0287381386593966
120 0.0278312582558544
125 0.0278450846737969
130 0.027444949918523
135 0.026559841581875
140 0.0265202623242549
145 0.0262059705186248
150 0.025904793814196
155 0.0249486771546611
160 0.0246108247797088
165 0.0245582097223337
170 0.0238183089121587
175 0.023206668562624
180 0.0238988319380041
185 0.0228571645179647
190 0.0227758184655382
195 0.0226671481619349
200 0.0224234573138848
};
\addplot [line width=1.5, green!50.1960784313725!black, dashed]
table {%
10 0.092503797992073
15 0.0784724435572367
20 0.0691743772234417
25 0.0623911262866822
30 0.0577950370672534
35 0.0545126691621726
40 0.053608022411036
45 0.0504354217590991
50 0.0496126004211145
55 0.0476532934373979
60 0.0466080076325164
65 0.0444021582388352
70 0.0438285798830125
75 0.0428234076321229
80 0.0424688674915211
85 0.041438637109511
90 0.0404585093112979
95 0.0408142819899935
100 0.0394265739103794
105 0.0384731591651674
110 0.0389216099900513
115 0.0367925916272637
120 0.0365413162993219
125 0.036132560254472
130 0.0358903427848754
135 0.0349860121733688
140 0.0348891672872669
145 0.0348981162814189
150 0.0356427464962294
155 0.0331425176203747
160 0.0327223680802673
165 0.0324823373400462
170 0.031192114832395
175 0.0296962025819959
180 0.0319059218864781
185 0.0300401290889394
190 0.0303854764507764
195 0.0305479270586749
200 0.029893490389142
};
\addplot [line width=1.5, green!50.1960784313725!black, dotted]
table {%
10 0.0673204608473013
15 0.0574741414870397
20 0.0505301120049928
25 0.0457420242821027
30 0.0413640354474453
35 0.0389006175854198
40 0.0373382625149208
45 0.0350220088849616
50 0.0332396863563785
55 0.0319122155627527
60 0.0301285420128167
65 0.0291987107958068
70 0.028293769999636
75 0.0270840622115013
80 0.0265429736632807
85 0.0260309533892869
90 0.0246813677700908
95 0.0239614259779331
100 0.0233830662192656
105 0.0231413246953888
110 0.0225983090497592
115 0.0222316370668324
120 0.0214777042653077
125 0.0213356950462972
130 0.020934297804121
135 0.0200878818032908
140 0.0203240479892552
145 0.0198065963803437
150 0.0193622297186961
155 0.0188769265692487
160 0.0187989958555088
165 0.0181824609081717
170 0.0178288381249388
175 0.0179333609328219
180 0.0178618363655721
185 0.0171801135576539
190 0.016917506629646
195 0.0169959287139908
200 0.0166375316731122
};
\end{axis}

\end{tikzpicture}}   \caption{Walsh model}
  \label{fig:sub1}
\end{subfigure}\hfill
\begin{subfigure}{.32\textwidth}
  \centering
  \resizebox{\textwidth}{!}{
\begin{tikzpicture}

\begin{axis}[
legend columns=2, 
legend cell align={left},
legend style={
	fill opacity=0.8,
	draw opacity=1,
	draw=none,
	text opacity=1,
	at={(0.5,1.3)},
	line width=1.5pt,
	anchor=north,
	/tikz/column 2/.style={
		column sep=10pt,
	}
},
legend entries={$m=10$,
	Mean,
	$m=25$,
	90th quantile,
	$m=50$,
	10th quantile},
tick align=outside,
tick pos=left,
x grid style={white!69.0196078431373!black},
xlabel={number of voters},
xmin=0.75, xmax=209.25,
xtick style={color=black},
y grid style={white!69.0196078431373!black},
ylabel={normalized positionwise distance},
ymin=0.0279687665474855, ymax=0.33349950615105,
ytick style={color=black},
ytick={0,0.05,0.1,0.15,0.2,0.25,0.3,0.35},
yticklabels={0.00,0.05,0.10,0.15,0.20,0.25,0.30,0.35}
]
\addlegendimage{red}
\addlegendimage{gray}
\addlegendimage{blue}
\addlegendimage{gray,dashed}
\addlegendimage{green!50.1960784313725!black}
\addlegendimage{gray,dotted}
\addplot [line width=1.5, red]
table {%
10 0.242816452893118
15 0.205219868222496
20 0.177653611350917
25 0.162371984183638
30 0.142889830532173
35 0.136042732181818
40 0.123882536940239
45 0.119700544133143
50 0.114101554924302
55 0.108098105762549
60 0.104476277667916
65 0.100055869832738
70 0.0968599346006346
75 0.0924498403515441
80 0.0900541343836283
85 0.0872709188762243
90 0.0844049709698072
95 0.0827164108230911
100 0.0805281644590161
105 0.0796673948944292
110 0.0757990222098776
115 0.0749711105461712
120 0.0732386104216931
125 0.0715545899429443
130 0.070884330472368
135 0.0697950521691919
140 0.0687109752219509
145 0.0675735605139323
150 0.0658010941333222
155 0.0642545181052816
160 0.0620797825543059
165 0.0622498379105871
170 0.0607877306949883
175 0.0605612652188148
180 0.0598974324493537
185 0.059813713425714
190 0.0592107093711431
195 0.0591696023583798
200 0.0581969450596212
};
\addplot [line width=1.5, red, dashed]
table {%
10 0.319611745259979
15 0.269693825935776
20 0.235009478190632
25 0.213395825005842
30 0.187679942731153
35 0.175981491595281
40 0.157343754505344
45 0.156702974258047
50 0.151262785008911
55 0.137110231026556
60 0.13351483465257
65 0.13070240729686
70 0.12304552712506
75 0.121085718539402
80 0.117024148489828
85 0.113094371083108
90 0.109211123187208
95 0.106846346773885
100 0.102269893698394
105 0.100776066938697
110 0.0990080488027271
115 0.0976166253881247
120 0.0928756335722
125 0.0930889251218601
130 0.0909708369396288
135 0.0904676178808917
140 0.0889174237212336
145 0.0872901215472005
150 0.0847013796115238
155 0.0815114848936598
160 0.0802296474406665
165 0.0790659339985613
170 0.0795474107476008
175 0.0770437658722089
180 0.0765919633393148
185 0.0767767163275769
190 0.0772349714640189
195 0.0769924649150309
200 0.077062968858941
};
\addplot [line width=1.5, red, dotted]
table {%
10 0.177130686830391
15 0.15190814771887
20 0.126500948597536
25 0.117942224871932
30 0.104100389564128
35 0.0983012117276138
40 0.0901515126679883
45 0.0865964457674912
50 0.0816567254788948
55 0.0786133243166136
60 0.0771062195470387
65 0.074100025489249
70 0.0708428159017455
75 0.068206277800103
80 0.0659398633845602
85 0.0641828227759988
90 0.0618731601615295
95 0.0615253619849682
100 0.059778878154854
105 0.060764123358284
110 0.0552713973306571
115 0.0553994809926459
120 0.0548421820569219
125 0.0522834360543074
130 0.0513578228906474
135 0.0520141530781984
140 0.0505675102358289
145 0.0495875731336348
150 0.0491870292067302
155 0.0480073090007698
160 0.0457125840242952
165 0.0464315075427294
170 0.0430591478111279
175 0.0454692949357471
180 0.0437108055942438
185 0.0430956297582298
190 0.042971614621241
195 0.0428090453684104
200 0.0418565274385566
};
\addplot [line width=1.5, blue]
table {%
10 0.227337797365557
15 0.188545458418532
20 0.162104245200037
25 0.150133346270662
30 0.13544797081351
35 0.125475016634386
40 0.117349617452023
45 0.110777344086345
50 0.105066820772078
55 0.099965384767737
60 0.0964738337810372
65 0.0925925918104118
70 0.0893788390462839
75 0.0858255884187141
80 0.0829345593317815
85 0.0816951443608298
90 0.0786744770660059
95 0.0769526794422358
100 0.0744741359640522
105 0.0744204056852509
110 0.0717817619109962
115 0.069936156300992
120 0.068517105519191
125 0.0675641670027443
130 0.0658155784402931
135 0.0650053487781686
140 0.0638848298206293
145 0.0622763905985854
150 0.0613649956046645
155 0.060691875756169
160 0.0590131279399942
165 0.0588465599903467
170 0.0577428733287954
175 0.0563025822323965
180 0.0560171161470778
185 0.0553142052780514
190 0.0548639481832521
195 0.0539658758207523
200 0.0532955380885948
};
\addplot [line width=1.5, blue, dashed]
table {%
10 0.271439259517222
15 0.225787757813501
20 0.196503133308859
25 0.179023155072131
30 0.162740666992389
35 0.148984612749281
40 0.141592377181559
45 0.131112869060598
50 0.125023473788483
55 0.119582658148227
60 0.114037526323675
65 0.109162563313461
70 0.105865909417984
75 0.101973943318276
80 0.0992884271106539
85 0.0971954309386022
90 0.093172167289035
95 0.0925408465626578
100 0.088763955689501
105 0.0885952703887597
110 0.0847566227621148
115 0.0833965006544111
120 0.0808595841369914
125 0.0803746190401188
130 0.079220521996747
135 0.0767684675851622
140 0.0748115932792783
145 0.0736785823730035
150 0.0722673375607253
155 0.0728895724748817
160 0.0700467206889208
165 0.0694809738797351
170 0.0689733474739254
175 0.0673430857653819
180 0.0667280241392571
185 0.0653087409472881
190 0.06532136628655
195 0.0649309914127046
200 0.0624565007169552
};
\addplot [line width=1.5, blue, dotted]
table {%
10 0.184664751026923
15 0.155085894984838
20 0.130552695596662
25 0.122478048908166
30 0.112056173554335
35 0.103288766403253
40 0.0958555192677662
45 0.0888151437285929
50 0.086300783804976
55 0.0821693165962083
60 0.0787878038164658
65 0.0763471276557539
70 0.0716005825625661
75 0.069967814131031
80 0.0674238982440259
85 0.0670157176370804
90 0.063881154189137
95 0.0631533878524071
100 0.061361761441311
105 0.0603854946653323
110 0.0583651461472842
115 0.0575230390473735
120 0.0567033125777156
125 0.0549799973745114
130 0.0540441751829348
135 0.0537581366088349
140 0.0532279364084108
145 0.0513057636294084
150 0.0504750706058425
155 0.0491075305548013
160 0.0482418719564941
165 0.0484967575677169
170 0.0481052293016826
175 0.0459078289595289
180 0.0455332759622252
185 0.0449968042693889
190 0.0451174035766878
195 0.0442818735802295
200 0.0445892167219426
};
\addplot [line width=1.5, green!50.1960784313725!black]
table {%
10 0.216183032506321
15 0.181701499434811
20 0.154476777705975
25 0.140448273150342
30 0.127158884935534
35 0.119711839194081
40 0.111511271405446
45 0.105024526580116
50 0.100148594216306
55 0.0952307219563561
60 0.0915424757796431
65 0.0876690148453713
70 0.0842080241062961
75 0.081595647309434
80 0.0794746063617377
85 0.0766310286038585
90 0.0748555172218389
95 0.0732648124664704
100 0.0711636328166871
105 0.0696003603668055
110 0.0681465362375884
115 0.0662639314871586
120 0.0648865752198031
125 0.0632015114993245
130 0.0618465100328456
135 0.0614605477924032
140 0.0595493167924488
145 0.0586633889204747
150 0.0577912660601168
155 0.0572947187141552
160 0.0564266356096693
165 0.0552564604821414
170 0.0545098975864973
175 0.0539111173690037
180 0.0526330377768162
185 0.0523555146478457
190 0.0516840416183829
195 0.0510013923438388
200 0.050742410240629
};
\addplot [line width=1.5, green!50.1960784313725!black, dashed]
table {%
10 0.248373319997781
15 0.204777550198107
20 0.178533753413932
25 0.161268121201899
30 0.146894839415702
35 0.137206580022549
40 0.127510650653374
45 0.120447052090803
50 0.114208068283003
55 0.107877754578085
60 0.103223479513333
65 0.0994455781122345
70 0.0963795648317358
75 0.0935136040936556
80 0.0901981604309485
85 0.0884782393960657
90 0.0856137885927441
95 0.0832937962489885
100 0.0806757944957685
105 0.0794553442898611
110 0.0783655349229556
115 0.075762828452058
120 0.0742483177175763
125 0.07218448909731
130 0.0713194280523965
135 0.0696945997403509
140 0.0676222067177854
145 0.0669728168545093
150 0.0656197386102104
155 0.0653437834285651
160 0.0643521979396075
165 0.0628025934505253
170 0.0615131222405063
175 0.0609214622349881
180 0.0598829747099015
185 0.0591399111837489
190 0.0581138224621864
195 0.0584958358119686
200 0.0579749212582199
};
\addplot [line width=1.5, green!50.1960784313725!black, dotted]
table {%
10 0.185004338629327
15 0.157673880360403
20 0.133635213740828
25 0.120389196544436
30 0.108665336610092
35 0.103365198778622
40 0.0948967950357066
45 0.0907688160537793
50 0.086839463604156
55 0.0829789770941582
60 0.0796314329846952
65 0.0764045526729164
70 0.0720338857313845
75 0.0703925681648749
80 0.0682894400883339
85 0.0650693626348723
90 0.0644522971139717
95 0.0640598755305642
100 0.0620513828512076
105 0.0599689020921617
110 0.0592155142345399
115 0.0570158963505278
120 0.0556488722226018
125 0.0540348776600525
130 0.0530584405848039
135 0.0531555474345193
140 0.0516636789357263
145 0.0508221191967779
150 0.0498730474563506
155 0.0492280340133682
160 0.0491184969957003
165 0.0478882485957218
170 0.0477968345942888
175 0.0469023139410968
180 0.0458009103222179
185 0.0454316896741767
190 0.0449831482353428
195 0.0440821891238746
200 0.0441083312060516
};
\end{axis}

\end{tikzpicture}}   \caption{GS/caterpillar model}
  \label{fig:sub2}
\end{subfigure}
\caption{For different vote distributions, behavior of the normalized positionwise distance between elections sampled from this distributions  
		 and the distribution's frequency matrix, for $10$/$25$/$50$ candidates and between $10$ and $200$ voters.}
\label{fig:vote_div1}
\end{figure*}

For each pair of matrices from the created set, we compute their normalized positionwise distance for $100$ candidates. Subsequently, for $m\in \{4,6,\dots, 98,100\}$ candidates, we compute the normalized positionwise distance of the frequency matrices of the two considered models for this number of candidates as well as the absolute and relative difference between their normalized distance for $m$ and $100$ candidates (where we normalize by their normalized distance for $100$ candidates). Finally, for each $m\in \{4,6,\dots, 98,100\}$, we take the maximum over the computed absolute/relative differences for all pairs of matrices.  In \Cref{fig:maxDiff}, we present these maxima for all considered values of $m$. Examining the maximum absolute difference (the blue line in \Cref{fig:maxDiff}), what stands out is that for $20$ or more candidates the normalized positionwise distance of any pair of considered vote distributions/matrices differs only by at most $0.0287$ from the pair's normalized positionwise distance for $100$ candidates (for $10$ or more candidates the error is at most $0.0547$). As the diameter of our space has at least length $1$, this change is quite small, and the considered vote distributions indeed remain at nearly the same distance for more than $20$ candidates. Considering the relative difference (the red line in \Cref{fig:maxDiff}), the picture appears to be a bit worse: for $20$ or more candidates, the normalized positionwise distance of any pair of considered vote distributions/matrices differs at most by $11.44$\% from their normalized positionwise distance for $100$ candidates. Nevertheless, this value is still relatively low, indicating an overall high robustness of the normalized positionwise distances of each pair of considered distributions with respect to the number of candidates.

\subsection{Skeleton Map for Different Number of Candidates}\label{app:maps}

After we have provided various arguments for why large parts of the skeleton map are presumably quite robust with respect to changing the number of candidates in the previous two subsections, in \Cref{fig:skeletonMapVaryingM}, we present the skeleton map for $5/10/25/50$ candidates. While the map for $5$ candidates looks a bit different from the other maps, the maps for $10$, $25$, and $50$ candidates differ only in that the frequency matrix for GS/caterpillar moves closer to $\AN$ and that the frequency matrix for the Walsh model moves closer to $\ID$ (both phenomena that we have already observed earlier).

\begin{figure*}
\centering
\begin{subfigure}{.45\textwidth}
  \centering
  \resizebox{0.711\textwidth}{!}{
\begin{tikzpicture}

\begin{axis}[
legend columns=2, 
legend cell align={left},
legend style={
	fill opacity=0.8,
	draw opacity=1,
	draw=none,
	text opacity=1,
	at={(0.5,1.3)},
	line width=1.5pt,
	anchor=north,
	/tikz/column 2/.style={
		column sep=10pt,
	}
},
legend entries={$m=10$,
	Mean,
	$m=25$,
	90th quantile,
	$m=50$,
	10th quantile},
tick align=outside,
tick pos=left,
x grid style={white!69.0196078431373!black},
xlabel={number of voters},
xmin=0.75, xmax=204.25,
xtick style={color=black},
y grid style={white!69.0196078431373!black},
ylabel={normalized positionwise distance},
ymin=0.0387121251936663, ymax=0.353712138430865,
ytick style={color=black},
ytick={0,0.05,0.1,0.15,0.2,0.25,0.3,0.35,0.4},
yticklabels={0.00,0.05,0.10,0.15,0.20,0.25,0.30,0.35,0.40}
]
\addlegendimage{red}
\addlegendimage{gray}
\addlegendimage{blue}
\addlegendimage{gray,dashed}
\addlegendimage{green!50.1960784313725!black}
\addlegendimage{gray,dotted}
\addplot [line width=1.5, red]
table {%
10 0.280343454872267
15 0.238063993308355
20 0.201737392305613
25 0.184323237168819
30 0.167387217683624
35 0.154660905465446
40 0.141002530745146
45 0.13687093981156
50 0.129278792234871
55 0.124130396129319
60 0.116321555192987
65 0.114267294819119
70 0.108542577708234
75 0.105837708652734
80 0.102728535737842
85 0.0998865186193525
90 0.0972963015571462
95 0.0951105774316297
100 0.0916535362618213
105 0.0899158280205471
110 0.0874609755540285
115 0.086079492101427
120 0.0839831640639088
125 0.0821288902892007
130 0.0806627860458361
135 0.0784077823779198
140 0.0767785025899759
145 0.0758836635979212
150 0.0747959613367313
155 0.0732961879875699
160 0.0724981054680591
165 0.0706651375378774
170 0.0700089147124402
175 0.0693073578210867
180 0.0680740770451122
185 0.0670859951845775
190 0.0665055831933789
195 0.0654268340696816
200 0.0645924243466421
};
\addplot [line width=1.5, red, dashed]
table {%
10 0.339393956010992
15 0.287070731886409
20 0.245454565571113
25 0.220606063396642
30 0.202020209237482
35 0.18450218767605
40 0.168181813993689
45 0.163636377250606
50 0.156484871019017
55 0.149311289462176
60 0.141414141598525
65 0.137062940261129
70 0.129870137538422
75 0.126101014334144
80 0.123484843269442
85 0.119429595181436
90 0.116498319714358
95 0.115470487884048
100 0.110909091049072
105 0.108542573813236
110 0.105234162554596
115 0.10197628745527
120 0.0999999935089639
125 0.0996606095044902
130 0.0974358988304933
135 0.0945005557753823
140 0.0917748996931495
145 0.0903030349562566
150 0.0901010166402116
155 0.0875855292108926
160 0.0871212013968916
165 0.0855831031885111
170 0.0837789711175543
175 0.0838095167821104
180 0.0808080844355352
185 0.0809336637457212
190 0.0800637947790551
195 0.0783216712262594
200 0.0773030329269893
};
\addplot [line width=1.5, red, dotted]
table {%
10 0.224242445961996
15 0.191919204344352
20 0.160606080327522
25 0.149090909754688
30 0.135151524209615
35 0.126406937702136
40 0.115151520028259
45 0.113131325206522
50 0.105454549238537
55 0.100826442918994
60 0.0939393990532015
65 0.0923076964463248
70 0.0883116957138885
75 0.08646464878411
80 0.0825757573838487
85 0.0816042797019084
90 0.0793939390530189
95 0.0781499138948592
100 0.0745454650485154
105 0.0724386773439068
110 0.0705234139480374
115 0.0714097497815436
120 0.0686868688825405
125 0.0669090937258619
130 0.0657342688942497
135 0.0628507365670168
140 0.0627705666603464
145 0.0616509884708759
150 0.0606060622542193
155 0.0588465328920971
160 0.0590909086167812
165 0.0580348945025242
170 0.0566844998209765
175 0.055913416747794
180 0.0548821585422212
185 0.055036855003599
190 0.0539074992140134
195 0.0535975100303238
200 0.053030307613539
};
\addplot [line width=1.5, blue]
table {%
10 0.299676303061692
15 0.243022676315278
20 0.2086724503207
25 0.184882048142131
30 0.17149617132692
35 0.158640580370452
40 0.148244161666252
45 0.139052717500286
50 0.13032980548003
55 0.125851693517373
60 0.120413417374739
65 0.115994732832545
70 0.111877298615261
75 0.105685478829382
80 0.104672320800571
85 0.101412829248032
90 0.0985938127559628
95 0.0960834037979545
100 0.0932259659558529
105 0.0911417659035532
110 0.0893533504998329
115 0.0871693984298662
120 0.0858390782546037
125 0.0835492342207223
130 0.0821550608583941
135 0.0801401942415246
140 0.0790414220612431
145 0.0776373579073697
150 0.075908658032407
155 0.0745665229924578
160 0.0742097956391248
165 0.072897418236154
170 0.0717865820878591
175 0.0704054050010223
180 0.0690369014972552
185 0.0689169267691278
190 0.0675258431156936
195 0.0665766654403444
200 0.0663139462802792
};
\addplot [line width=1.5, blue, dashed]
table {%
10 0.33598079513042
15 0.270256434757119
20 0.231548095281379
25 0.206923077148027
30 0.190076938601067
35 0.178681329071808
40 0.165302890990503
45 0.155047021875194
50 0.145980767903921
55 0.139276225872947
60 0.1347500063667
65 0.129766284695003
70 0.1244011094928
75 0.11834616791636
80 0.117216354343467
85 0.112972856688206
90 0.109596171415447
95 0.107726724793275
100 0.104038461852962
105 0.100866311793932
110 0.0994493023707316
115 0.0979949840263893
120 0.0955657110970396
125 0.0939307709126017
130 0.0916582971175488
135 0.0904515723608291
140 0.0878090780720903
145 0.0873766640469862
150 0.0849423099434576
155 0.0831861047038379
160 0.0830901459289285
165 0.0818846137674812
170 0.0798834886186971
175 0.0787362648102527
180 0.0771848392332546
185 0.0773825364769436
190 0.0750334005507354
195 0.0738915252185856
200 0.0741394263969806
};
\addplot [line width=1.5, blue, dotted]
table {%
10 0.267307724388173
15 0.215230800677091
20 0.185769240835753
25 0.164192303112493
30 0.153256425334929
35 0.140208797389641
40 0.130764430085233
45 0.124307695131462
50 0.114615385565692
55 0.112482521209257
60 0.106599365393273
65 0.103159772734552
70 0.100071433980842
75 0.0935512917223745
80 0.0932163477848427
85 0.0905588284922907
90 0.0882457324751438
95 0.0848178167442361
100 0.0827788527472876
105 0.0810989068391231
110 0.0797884575485324
115 0.07730602071925
120 0.0757516058323045
125 0.0741461582437086
130 0.0728964592249563
135 0.0713005697847201
140 0.070493135565462
145 0.0689164453186095
150 0.0673717989290778
155 0.0664255555790777
160 0.0656262016407429
165 0.0643449889746709
170 0.0640735311030697
175 0.0623626404925464
180 0.0616453058623637
185 0.0609958414534608
190 0.0603350190678611
195 0.059427030378272
200 0.0591250057565048
};
\addplot [line width=1.5, green!50.1960784313725!black]
table {%
10 0.297842436178202
15 0.243080092964753
20 0.210054279344858
25 0.186851151043064
30 0.171351765162875
35 0.158906970713537
40 0.14845754508961
45 0.140540113548404
50 0.131547981124318
55 0.126583955999819
60 0.121433122082117
65 0.116795555483073
70 0.112586362310178
75 0.108003963737608
80 0.104943883533756
85 0.101295484864759
90 0.0990892300000653
95 0.0960074100399665
100 0.0933361352791823
105 0.0919950995702433
110 0.0892895927150067
115 0.0877572609177388
120 0.0857597746151653
125 0.0839081816542104
130 0.0820229190979384
135 0.0806400131485986
140 0.0793533123728214
145 0.0781256892445873
150 0.076602782743134
155 0.0755335914980428
160 0.0740423686391323
165 0.073191411052728
170 0.0720501717930811
175 0.0710509473073981
180 0.0696332235982595
185 0.0694700195157409
190 0.0680812538689913
195 0.0670762661804876
200 0.066393661267745
};
\addplot [line width=1.5, green!50.1960784313725!black, dashed]
table {%
10 0.321978397646389
15 0.261999231735233
20 0.228115255331524
25 0.202746711924028
30 0.183758334167937
35 0.17158773268024
40 0.160092464726673
45 0.151517692009727
50 0.142333734852551
55 0.136763947374117
60 0.131196498213445
65 0.125944053649115
70 0.121428242907608
75 0.116896378240209
80 0.113024027483808
85 0.109564309292288
90 0.108292391959305
95 0.10289379762148
100 0.10110444278643
105 0.0996597613760305
110 0.0967146215752429
115 0.0949417023504732
120 0.0923469485903979
125 0.0910588382040055
130 0.0884151986765475
135 0.0870805227201061
140 0.0863814197306265
145 0.0839584456468258
150 0.0831884820808043
155 0.081700352221137
160 0.0797428030357445
165 0.0788935261408464
170 0.0782052209683541
175 0.0766715849755334
180 0.0752060964507904
185 0.0743679996573177
190 0.0736574238184102
195 0.0721308909736726
200 0.0718607473830018
};
\addplot [line width=1.5, green!50.1960784313725!black, dotted]
table {%
10 0.275822335126556
15 0.223575869863596
20 0.193738298164625
25 0.171654267539801
30 0.158791542203486
35 0.145437156521052
40 0.137346956819868
45 0.129824741389759
50 0.12114765821349
55 0.116340507022446
60 0.11233615063104
65 0.107993362168474
70 0.103906035903251
75 0.0996846782280599
80 0.0965270296346266
85 0.0936159973382717
90 0.0913400138720494
95 0.0889217244394097
100 0.0860240082335876
105 0.0845167886548522
110 0.0822638927150567
115 0.0810198860275386
120 0.0790892532810855
125 0.0772331416004283
130 0.0759084082641253
135 0.0748073483032419
140 0.0728598955455197
145 0.0719953662072983
150 0.0707963255407358
155 0.0699261935894005
160 0.0682749155367694
165 0.0676143955946917
170 0.0663608625269577
175 0.0653267015659419
180 0.0640066779649057
185 0.0644268544914038
190 0.0626190739792256
195 0.0619992038275318
200 0.0607779156813557
};
\end{axis}

\end{tikzpicture}}   \caption{Impartial Culture}
  \label{fig:sub1}
\end{subfigure}\hfill
\begin{subfigure}{.45\textwidth}
  \centering
  \resizebox{0.711\textwidth}{!}{
\begin{tikzpicture}

\begin{axis}[
legend columns=2, 
legend cell align={left},
legend style={
	fill opacity=0.8,
	draw opacity=1,
	draw=none,
	text opacity=1,
	at={(0.5,1.3)},
	line width=1.5pt,
	anchor=north,
	/tikz/column 2/.style={
		column sep=10pt,
	}
},
legend entries={$m=16$,
	Mean,
	$m=32$,
	90th quantile,
	$m=64$,
	10th quantile},
tick align=outside,
tick pos=left,
x grid style={white!69.0196078431373!black},
xlabel={number of voters},
xmin=0.75, xmax=204.25,
xtick style={color=black},
y grid style={white!69.0196078431373!black},
ylabel={normalized positionwise distance},
ymin=0.020622295129276, ymax=0.494684741888111,
ytick style={color=black},
ytick={0,0.1,0.2,0.3,0.4,0.5},
yticklabels={0.0,0.1,0.2,0.3,0.4,0.5}
]
\addlegendimage{red}
\addlegendimage{gray}
\addlegendimage{blue}
\addlegendimage{gray,dashed}
\addlegendimage{green!50.1960784313725!black}
\addlegendimage{gray,dotted}
\addplot [line width=1.5, red]
table {%
10 0.302517159250291
15 0.241568141000382
20 0.207801965249885
25 0.18752244533949
30 0.169225489305442
35 0.156597271100928
40 0.145615199468226
45 0.136399727652689
50 0.133385586716454
55 0.128009802899378
60 0.121614870092083
65 0.116553204821678
70 0.109238376056301
75 0.10664042474198
80 0.102086274516641
85 0.101051930427478
90 0.101267753950959
95 0.096126083148899
100 0.093682352864567
105 0.0937792473315608
110 0.0880337336058272
115 0.0855586721588876
120 0.08394020041838
125 0.0831690581996797
130 0.0818352940942289
135 0.0795877820892515
140 0.0770052537674851
145 0.0776310207905869
150 0.0758728418470013
155 0.073874019580948
160 0.0737656861865053
165 0.072539349433619
170 0.0680309694724924
175 0.0719085580701483
180 0.0689004913163828
185 0.0691209181204584
190 0.0664017798922518
195 0.0673116276728756
200 0.0655877454798742
};
\addplot [line width=1.5, red, dashed]
table {%
10 0.470411775594249
15 0.367745091827477
20 0.31852941107224
25 0.26759998361854
30 0.249156869961058
35 0.233478987545652
40 0.212647057223846
45 0.198830056115985
50 0.200252935899531
55 0.197096256306943
60 0.177627456065486
65 0.169004523692762
70 0.161668061929591
75 0.157862742569517
80 0.157970587151016
85 0.157455005347729
90 0.148049017410068
95 0.139304955465829
100 0.137088225624141
105 0.139920160183135
110 0.128074863364591
115 0.124363174114157
120 0.125539218621219
125 0.121602353887523
130 0.119710400192177
135 0.121222224761458
140 0.112697481464814
145 0.112516223446411
150 0.116576468203874
155 0.108160346856012
160 0.112132351280135
165 0.107102499950458
170 0.0999861569746452
175 0.107609244076645
180 0.103728754524799
185 0.0993155794020961
190 0.0990108343389104
195 0.0995000067747691
200 0.0979529440928908
};
\addplot [line width=1.5, red, dotted]
table {%
10 0.199705885920454
15 0.157421563951408
20 0.133794122750268
25 0.123317648785956
30 0.109313729558797
35 0.102117648365743
40 0.0911764750585836
45 0.0884803993211073
50 0.0848764697199359
55 0.0806737935630714
60 0.080284315054031
65 0.0757330313061967
70 0.0696722733886803
75 0.0714254909271703
80 0.0646764719135621
85 0.0636712837394546
90 0.0651209130927044
95 0.064701239913702
100 0.059923529638087
105 0.0601176464995917
110 0.0559812808255939
115 0.0563491037487984
120 0.052980395472225
125 0.0533882335398127
130 0.0521357509406174
135 0.0524629618446617
140 0.0491806773359285
145 0.0500791087834274
150 0.0488686279207468
155 0.0493690688995754
160 0.0460147027890472
165 0.046926021299818
170 0.0449394503411125
175 0.0468924349108163
180 0.0454215731121161
185 0.0440119204056614
190 0.0432972110983203
195 0.0437096565041472
200 0.0421705881637685
};
\addplot [line width=1.5, blue]
table {%
10 0.297222998742695
15 0.245325318477895
20 0.208932797938498
25 0.184537663274861
30 0.172110683092791
35 0.158947614820643
40 0.148564093703257
45 0.138609564258969
50 0.129307229085001
55 0.124836844715699
60 0.119329645311831
65 0.118539418253095
70 0.111639897345109
75 0.108919465288307
80 0.105272486908316
85 0.100093525871607
90 0.0997090359259568
95 0.097037511201701
100 0.0933271981479886
105 0.0925393992049286
110 0.0905521910761156
115 0.087351498290667
120 0.0863459601005374
125 0.0831511068482859
130 0.0857474060121283
135 0.080664919643733
140 0.0818349872693376
145 0.0789542882258431
150 0.0769952183042108
155 0.0747605457615457
160 0.0744739120741459
165 0.0740591412281685
170 0.0726194422053608
175 0.0693385731534648
180 0.0686469272700278
185 0.066617377177245
190 0.0680895803389363
195 0.0676232700873319
200 0.0658740583867315
};
\addplot [line width=1.5, blue, dashed]
table {%
10 0.470168625456322
15 0.372823814012088
20 0.318185495056452
25 0.26222872370499
30 0.257699173198374
35 0.236746966236345
40 0.214486821801497
45 0.195892387747633
50 0.189642947808432
55 0.181237991199531
60 0.174987784820633
65 0.16790660453354
70 0.166914007848243
75 0.160384657426227
80 0.155799131757987
85 0.148914304746539
90 0.145711950261878
95 0.143223680383212
100 0.134129035904392
105 0.139187428297317
110 0.133996935561299
115 0.126942816888504
120 0.127853129417403
125 0.120416423313516
130 0.133305034594736
135 0.115229849474193
140 0.122852422342338
145 0.113573668555212
150 0.11470331919934
155 0.106077716025985
160 0.109244865862396
165 0.10991613624175
170 0.105655710030992
175 0.101726955323428
180 0.0974617006568528
185 0.0957619668593592
190 0.100218396840885
195 0.0980556721326587
200 0.0957917878611827
};
\addplot [line width=1.5, blue, dotted]
table {%
10 0.198049856335624
15 0.159707975496819
20 0.137280060466696
25 0.12333650485424
30 0.114075028320605
35 0.10336719561106
40 0.0946480965187167
45 0.0915033405810551
50 0.0861290336212367
55 0.0828435715692134
60 0.0781573778548199
65 0.0768945941264547
70 0.0732933604133881
75 0.0716280604212865
80 0.0694831382554254
85 0.0638118849133204
90 0.0639397987844614
95 0.0631409124987435
100 0.0608812312997025
105 0.0598533708290839
110 0.057685285325427
115 0.058661224055063
120 0.056169350543592
125 0.0550299121716942
130 0.0545888785742595
135 0.0536542855629211
140 0.0518323201292858
145 0.0511220809869752
150 0.0499912048091899
155 0.0506406652929584
160 0.0488159829930913
165 0.0486471391689917
170 0.0468697167823741
175 0.0457860273890787
180 0.0456822289197632
185 0.0444639177408474
190 0.0439745712479101
195 0.0436183183560492
200 0.042884897063773
};
\addplot [line width=1.5, green!50.1960784313725!black]
table {%
10 0.304046828047669
15 0.245706319545224
20 0.209848345279744
25 0.185896638724612
30 0.172207270699215
35 0.157891873399919
40 0.147335811336001
45 0.141027236708521
50 0.132962225090558
55 0.127185705524063
60 0.123359144915468
65 0.114327392743983
70 0.111878007491424
75 0.107820226809851
80 0.106005189662172
85 0.103714977381646
90 0.0982465420659286
95 0.0959042558195581
100 0.093637674900586
105 0.0921250746119983
110 0.0893366811775337
115 0.0877285044658383
120 0.0847464110344792
125 0.0857996787526844
130 0.0832299663038431
135 0.0807704751612357
140 0.0791189487730579
145 0.0763019334907861
150 0.0759165877279847
155 0.0750940927266381
160 0.075573959304502
165 0.0736403022370527
170 0.0725386650873057
175 0.0715962817069252
180 0.0701638546367602
185 0.0672912561503104
190 0.0684967595762078
195 0.0679205508433608
200 0.0664280210960343
};
\addplot [line width=1.5, green!50.1960784313725!black, dashed]
table {%
10 0.473136448853618
15 0.373305413834341
20 0.319224363487664
25 0.256121324088438
30 0.263085776237232
35 0.240499400757526
40 0.211333349873977
45 0.203650184651161
50 0.199936708516472
55 0.18674196271833
60 0.185853933368847
65 0.16710583216364
70 0.161295782015233
75 0.158996131476629
80 0.157573281581626
85 0.149525609369639
90 0.143889851165328
95 0.140909855678923
100 0.134235165134864
105 0.135562769912675
110 0.128448887697179
115 0.129005974580367
120 0.126036021085831
125 0.12516756484708
130 0.122609072049786
135 0.116224858003631
140 0.114965917055151
145 0.113256787588417
150 0.110696886620611
155 0.106754713725885
160 0.110340209574517
165 0.108309627011446
170 0.106031855122074
175 0.106249924242469
180 0.102399617768892
185 0.0994814830403025
190 0.100448564156931
195 0.100548883627336
200 0.0963228909320611
};
\addplot [line width=1.5, green!50.1960784313725!black, dotted]
table {%
10 0.209750004444585
15 0.160492534175133
20 0.141727104499499
25 0.124233511269857
30 0.112981689658303
35 0.104925887415343
40 0.0974349831310766
45 0.0932475563822376
50 0.088129575882148
55 0.0853494811396459
60 0.0800317497361765
65 0.0763149449452348
70 0.0738499498564491
75 0.0698223161816106
80 0.0679258292186118
85 0.0671859493026094
90 0.0650795144768561
95 0.0644357514516502
100 0.0618921215274122
105 0.0593717507730774
110 0.0591277454176665
115 0.0577176073845135
120 0.0568444745316965
125 0.0549302766377263
130 0.0548345293402126
135 0.0531437068118717
140 0.0525432374393192
145 0.0503040920118818
150 0.0504400815189758
155 0.050651264732956
160 0.0482495435773508
165 0.0477097483810324
170 0.0479551778513351
175 0.046590263126111
180 0.0455790071794593
185 0.044315835064202
190 0.0460579811040681
195 0.04395882035841
200 0.0440448716201652
};
\end{axis}

\end{tikzpicture}}   \caption{GS/balanced model}
  \label{fig:sub2}
\end{subfigure}
\caption{For different vote distributions, behavior of the normalized positionwise distance between elections sampled from this distributions  
		 and the distribution's frequency matrix, for $10$/$25$/$50$ candidates and between $10$ and $200$ voters.}
\label{fig:vote_div2}
\end{figure*}
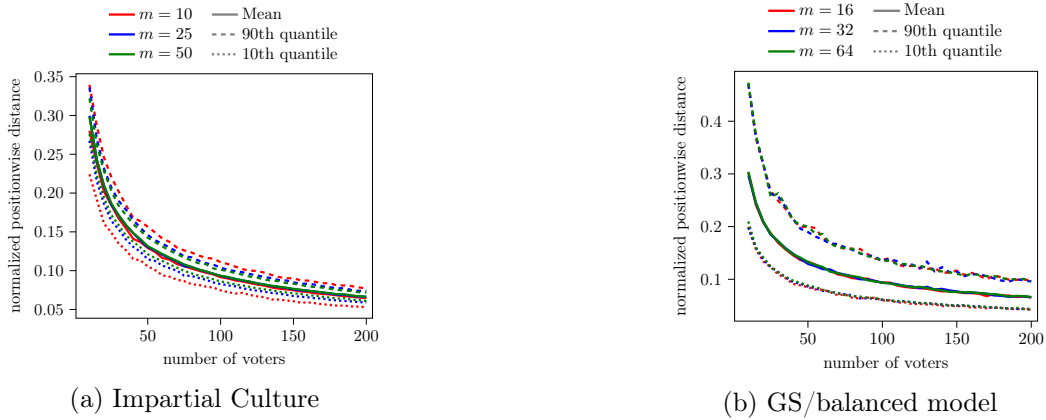

\subsection{Variance of a Vote Distribution} \label{sec:variance} \label{app:variance}
As argued in Section~\ref{sec:setup}, a frequency matrix of a distribution can be interpreted as a matrix of
an ``ideal'' election sampled from this distribution. 
In this section we ask how far, on average,
the elections sampled from our distributions land from the ``ideal''
ones.  This distance may also serve as a measure of ``diversity'' for elections sampled from a given distribution.

For each of our distributions, we consider elections with $10$, $25$,
and $50$ candidates, and vary the number of voters from $10$ to $200$ (with a
step of $5$). For each combination of these parameters, we sample
$600$ elections and, for each election, compute the positionwise distance between its
frequency matrix and the matrix of the respective
distribution.\footnote{For the GS/balanced distribution we consider
  $16$, $32$, and $64$ candidates, as this model requires the number of
  candidates to be a power of two and we do not consider GS/flat trees, as this distribution is too simple.}  We show the results in \Cref{fig:vote_div1,fig:vote_div2}. As expected, for all vote distributions,
increasing the number of votes decreases the average distance of an election from the
distribution's matrix (indeed, in the limit this distance is zero).  What is
more surprising, this distance does not seem to depend on the number
of candidates for the Conitzer, $\IC$, and GS/balanced distribution.
For the Walsh distribution (and, to a lesser extent, for GS/caterpillar), the sampled
elections get slightly closer to the respective matrix as we increase the
number of candidates.  Moreover, if we fix the number of candidates
and voters, then for all our distributions the elections sampled from
them are, roughly, at the same distance from the distribution's
matrix.  For an illustration of this effect consider \Cref{fig:CitySizeCon} for the Conitzer distribution; there, the 10th quantile (dotted) and the 90th
quantile (dashed) only differ by a factor of four. Lastly, comparing the plots for IC and GS/balanced (\Cref{fig:vote_div2}), who both have $\UN$ as their frequency matrix, the average distance of elections sampled form one of these two models to $\UN$ (which is their frequency matrix) is the same for both distributions. However, for IC, the 10th and 90th quantile of the distances of elections to $\UN$ are closer to the average than for GS/balanced, which indicates that IC produces in some sense less varied elections than GS/balanced.  

In \Cref{fig:CitySizeComp} we compare the average distances between
elections sampled from various vote distributions and the distribution's matrices (we fix the number of candidates to $50$ and vary the
number of voters).  While, on average, $\IC$ elections and GS/balanced elections end up at nearly the same
distance from $\UN$ (which is their frequency matrix), Conitzer
elections and GS/caterpillar elections end up closer to their
distribution's matrices, and for Walsh this effect is considerably
stronger. Overall, it is remarkable that even for $200$ voters, for
the Conitzer, $\IC$, GS/balanced, and GS/caterpillar, the
average distance of a sampled election from the respective matrix is
still above $0.05$ (so at least $5\%$ of the diameter of the whole
space).  We also performed the same experiment for the Mallows model
with different values of the normalized dispersion parameter (see \Cref{fig:ovMallows}): For a varying number of voters, we depict the average distance of $600$ elections with $50$ candidates sampled from Mallows model for different values of $\normphi$ to the distribution's frequency matrix. Quite intuitively, the more swaps we make to the central vote (i.e., the higher $\normphi$ is), the higher is the average distance of a sampled elections from the distribution's frequency matrix. 

It is interesting to contrast the data from \Cref{fig:CitySizeCon}
with the map of \citet{boe-bre-fal-nie-szu:t:compass} in
\Cref{fig:compass_map}. \citeauthor{boe-bre-fal-nie-szu:t:compass}
considered $10$ candidates and $100$ voters. For these parameters, in
\Cref{fig:CitySizeCon} we see that $10\%$ of the elections are still
farther from the Conitzer matrix than about $12\%$ of the distance
from $\UN$ to $\ID$. This is roughly reflected by the size of the area
taken by Conitzer elections in Figure~\ref{fig:compass_map}.  Similar
observations hold for the other distributions too.  While this might
be a coincidence, it confirms the value of their map.

\begin{figure}[t]
    \centering
    \begin{minipage}{.49\textwidth}
        \centering
        \centering
  \resizebox{0.66\textwidth}{!}{
\begin{tikzpicture}

\definecolor{color0}{rgb}{0.12156862745098,0.466666666666667,0.705882352941177}
\definecolor{color1}{rgb}{1,0.498039215686275,0.0549019607843137}
\definecolor{color2}{rgb}{0.172549019607843,0.627450980392157,0.172549019607843}
\definecolor{color3}{rgb}{0.83921568627451,0.152941176470588,0.156862745098039}
\definecolor{color4}{rgb}{0.580392156862745,0.403921568627451,0.741176470588235}

\begin{axis}[
legend columns=3, 
legend cell align={left},
legend style={
	fill opacity=0.8,
	draw opacity=1,
	draw=none,
	text opacity=1,
	at={(0.5,1.2)},
	line width=1.5pt,
	anchor=north,
	/tikz/every even column/.append style={column sep=0.5cm}
},
tick align=outside,
tick pos=left,
x grid style={white!69.0196078431373!black},
xlabel={number of voters},
xmin=0.5, xmax=209.5,
xtick style={color=black},
y grid style={white!69.0196078431373!black},
ylabel={normalized positionwise distance},
ymin=0.00164411295773807, ymax=0.34,
ytick style={color=black},
ytick={0,0.05,0.1,0.15,0.2,0.25,0.3,0.35},
yticklabels={0.00,0.05,0.10,0.15,0.20,0.25,0.30,0.35}
]
\addplot [line width=1.5, color0]
table {%
10 0.221387073256287
15 0.182818240701393
20 0.159333732076431
25 0.143854319222008
30 0.135805323714332
35 0.123627500870221
40 0.113402312374042
45 0.108772748153759
50 0.102845231466984
55 0.102310208097464
60 0.0933441538630831
65 0.0886422162240396
70 0.087423948179446
75 0.0813223354329056
80 0.0816625131981741
85 0.0797020600370041
90 0.0764794355832177
95 0.0755563771639581
100 0.0731934076362923
105 0.0695475648127484
110 0.0693505064357074
115 0.0703941976874821
120 0.0647249403509929
125 0.0647985852128152
130 0.0642898425135902
135 0.0623343416288228
140 0.0613019252621404
145 0.0608734217149486
150 0.059961343099141
155 0.0584536918274573
160 0.0576554348415044
165 0.0575095401973534
170 0.0553882246801485
175 0.0551381682623912
180 0.0541692720000696
185 0.0532790229533825
190 0.0529004816104378
195 0.0531293110997996
200 0.0519895667773506
};
\addlegendentry{Conitzer}
\addplot [line width=1.5, color1]
table {%
10 0.0783111387154496
15 0.0668481534106084
20 0.0593117227357779
25 0.0539512990739728
30 0.0503018475763901
35 0.0470628264107654
40 0.0452345758407025
45 0.0424416714498778
50 0.041605174045807
55 0.0393594684570455
60 0.0382941189504426
65 0.0365982825087806
70 0.0363994514409751
75 0.0353255681586829
80 0.0333408183785227
85 0.0327254785505016
90 0.0321205730502011
95 0.0313940990492667
100 0.0303685408278859
105 0.0297037025734351
110 0.029190323045956
115 0.0290476701290367
120 0.0279228978214751
125 0.0277178439037137
130 0.0270390679277027
135 0.0267086740382049
140 0.026263331382236
145 0.0259162180531499
150 0.0249209302734767
155 0.0251490689952684
160 0.0249305299133989
165 0.0246046492488304
170 0.0243368529004318
175 0.0233445226942989
180 0.0232584579025807
185 0.0230207311931128
190 0.0225945051290265
195 0.0221309805392548
200 0.0223297696904587
};
\addlegendentry{Walsh}
\addplot [line width=1, color3]
table {%
10 0.298691814912445
15 0.243148387293881
20 0.210419865576196
25 0.186898049788118
30 0.171894355115755
35 0.158665948493951
40 0.148645200119572
45 0.139903850230586
50 0.131803603559418
55 0.126843055986198
60 0.120925879467099
65 0.116162736091676
70 0.111876798203174
75 0.1085650418224
80 0.105008519434395
85 0.102293446415394
90 0.0993386630868657
95 0.0960711280691127
100 0.0932004809612613
105 0.0913548469559333
110 0.0894965628868084
115 0.0876853783899734
120 0.0855515048482834
125 0.083723179550067
130 0.0825151279340444
135 0.0807424481460216
140 0.079906213649559
145 0.0776619893356132
150 0.0760461598496065
155 0.075236300950195
160 0.0738316343860383
165 0.0730127868402403
170 0.0722465186436441
175 0.070848695172525
180 0.0699625333283703
185 0.069163559519836
190 0.0679461767239718
195 0.0671461711980279
200 0.0659084272017463
};
\addlegendentry{IC}
\addplot [line width=1, color4]
table {%
10 0.303432282615521
15 0.24028358802002
20 0.209550985778171
25 0.186187611177864
30 0.173414494404058
35 0.159613546139182
40 0.145560917788591
45 0.139639298664306
50 0.136778450758311
55 0.131232314051512
60 0.120668439523357
65 0.116154668225984
70 0.114015730809094
75 0.112604316059245
80 0.106617117496506
85 0.100965418245596
90 0.0971377012601068
95 0.0974373199364537
100 0.0934012556985835
105 0.0911397397940437
110 0.0893572682695263
115 0.0879477087559522
120 0.0873055462766898
125 0.0833179667276377
130 0.0820276896387295
135 0.0789149647638263
140 0.0808399556271893
145 0.0772099845328656
150 0.0772904203375972
155 0.0746283383714183
160 0.0732606349657125
165 0.0728312023086758
170 0.0728674639054041
175 0.0709821365141368
180 0.0711383553152911
185 0.0661268315015313
190 0.0664628307161607
195 0.0657450254120806
200 0.0673775769946869
};
\addlegendentry{GS/balanced}

\addplot [line width=1, yellow]
table {%
10 0.217031373501329
15 0.181048956431259
20 0.155424904998432
25 0.142055082424457
30 0.127058900991207
35 0.118720369096896
40 0.111220605328009
45 0.105043933946032
50 0.0996835721498994
55 0.0952491255696447
60 0.0911081781060422
65 0.087695859728143
70 0.0846200173783017
75 0.0822516694707565
80 0.0795303986477828
85 0.0771953407927547
90 0.0741843811628104
95 0.0730909068206869
100 0.0710967293574201
105 0.069407171786522
110 0.0677836577545885
115 0.0662235511602864
120 0.0645815770496247
125 0.0637983606504288
130 0.0620648335496755
135 0.0609115729595656
140 0.0607737058706915
145 0.0594068856957251
150 0.0579069137249577
155 0.0571570431372999
160 0.0564061791307392
165 0.0554903729855772
170 0.0542221551990441
175 0.0535994014316089
180 0.0529886552534269
185 0.0522712605296601
190 0.0517654245499715
195 0.0512227791817313
200 0.0505694798959432
};
\addlegendentry{GS/caterpillar}
\end{axis}

\end{tikzpicture}}   \caption{Average normalized positionwise distance between elections sampled from various voter distributions and the frequency matrices of the respective models, for $50$ candidates ($64$ for GS) and between $10$ and $200$ voters.} \label{fig:CitySizeComp}
    \end{minipage}\hfill
    \begin{minipage}{0.49\textwidth}
        \centering
        \resizebox{0.66\textwidth}{!}{
\begin{tikzpicture}

\definecolor{color0}{rgb}{0.12156862745098,0.466666666666667,0.705882352941177}
\definecolor{color1}{rgb}{1,0.498039215686275,0.0549019607843137}
\definecolor{color2}{rgb}{0.172549019607843,0.627450980392157,0.172549019607843}
\definecolor{color3}{rgb}{0.83921568627451,0.152941176470588,0.156862745098039}
\definecolor{color4}{rgb}{0.580392156862745,0.403921568627451,0.741176470588235}
\definecolor{color5}{rgb}{0.549019607843137,0.337254901960784,0.294117647058824}
\definecolor{color6}{rgb}{0.890196078431372,0.466666666666667,0.76078431372549}
\definecolor{color7}{rgb}{0.737254901960784,0.741176470588235,0.133333333333333}

\begin{axis}[
legend columns=3, 
legend cell align={left},
legend style={
	fill opacity=0.8,
	draw opacity=1,
	draw=none,
	text opacity=1,
	at={(0.5,1.3)},
	line width=1.5pt,
	anchor=north,
	/tikz/column 2/.style={
		column sep=10pt,
	}
},
tick align=outside,
tick pos=left,
x grid style={white!69.0196078431373!black},
xlabel={n},
xmin=-7.75, xmax=206.75,
xtick style={color=black},
y grid style={white!69.0196078431373!black},
ylabel={normalized mean distance},
ymin=-0.0118653661635964, ymax=0.648038818073075,
ytick style={color=black},
ytick={-0.1,0,0.1,0.2,0.3,0.4,0.5,0.6,0.7},
yticklabels={−0.1,0.0,0.1,0.2,0.3,0.4,0.5,0.6,0.7}
]
\addplot [semithick, color0]
table {%
2 0.102958617908132
7 0.0637954569095837
12 0.0518069161140141
17 0.0460331067223261
22 0.0416825770542254
27 0.0387810728606994
32 0.0354644198728675
37 0.0332255097806207
42 0.0318312156888906
47 0.0303068861077274
52 0.0289898220925436
57 0.0278973505937448
62 0.0267486677485092
67 0.0260900984936955
72 0.0250658895475686
77 0.0242952467985139
82 0.0235743560253912
87 0.023104491375613
92 0.0223662826826373
97 0.0218468031985149
102 0.0214228936038138
107 0.0209394091136606
112 0.0203707903812749
117 0.0201797112322131
122 0.0195802283474592
127 0.0194178278720804
132 0.018936281440742
137 0.0187034718182563
142 0.0183409790493973
147 0.0178911856663451
152 0.0176339245336827
157 0.0174037360807354
162 0.0172707744457068
167 0.0168362545281652
172 0.0166629846888578
177 0.0165161033826332
182 0.0160671026440715
187 0.0159378786350777
192 0.0156385724859941
197 0.0155708556680352
};
\addlegendentry{$0.1$}
\addplot [semithick, color1]
table {%
2 0.179891517523228
7 0.109152571765145
12 0.0876733130720241
17 0.0747607029666167
22 0.0676056719444268
27 0.0622476197266191
32 0.0580619110861142
37 0.0544525903048096
42 0.0515697856903219
47 0.0492361849423552
52 0.0469038331730718
57 0.0452677753097774
62 0.0437459989885039
67 0.0422176124653244
72 0.0413197110356146
77 0.0398100586845437
82 0.0387202366963564
87 0.0378801604581007
92 0.0368091108044779
97 0.036433580241969
102 0.0353950743704256
107 0.0349867427249725
112 0.0341024244858973
117 0.0333800531630634
122 0.0326747768812236
127 0.032116587522699
132 0.0316041087659111
137 0.0308432412822627
142 0.0309230272949432
147 0.0301921129620237
152 0.029968923269973
157 0.0292614048168571
162 0.0289781248835905
167 0.0287670782793538
172 0.0284442573175732
177 0.0279418991425311
182 0.0273888387257239
187 0.0272290275417148
192 0.0267763641440927
197 0.0262395362811485
};
\addlegendentry{$0.2$}
\addplot [semithick, color2]
table {%
2 0.251190353014416
7 0.145970630409224
12 0.115690530506328
17 0.0992403536840826
22 0.0885669857689185
27 0.0816283966913659
32 0.0756478731507643
37 0.0714057445822693
42 0.0670144949914322
47 0.0641619993173357
52 0.0612821892884116
57 0.0593697215031726
62 0.0570331090935232
67 0.0548421397436189
72 0.0533676156899174
77 0.0515523572518551
82 0.0504160440370474
87 0.0493479550207747
92 0.048229903264998
97 0.0469563809394208
102 0.0458777836344634
107 0.0447469007955386
112 0.0442590321542453
117 0.0434904951297332
122 0.0427329141947311
127 0.0422251908993349
132 0.0412788832361437
137 0.040411661481007
142 0.0398819430349972
147 0.0394830107727752
152 0.0387233760082324
157 0.0378320619085183
162 0.0377931589180405
167 0.0374529631278153
172 0.0368470687444917
177 0.0362470943449025
182 0.0358198468518308
187 0.0353867946659726
192 0.0351668797328048
197 0.0347408803088912
};
\addlegendentry{$0.3$}
\addplot [semithick, color3]
table {%
2 0.318104219373831
7 0.179774899244665
12 0.140455733954927
17 0.120399688824225
22 0.106016721017227
27 0.0970483235258755
32 0.0890190734152915
37 0.0838269048253304
42 0.0795028994018984
47 0.0753965984065533
52 0.0722787559702651
57 0.0697711454281734
62 0.0666304896539254
67 0.0646118183326684
72 0.0625045666789618
77 0.0609216786927168
82 0.0590986856290905
87 0.0576560535842524
92 0.0561397335936141
97 0.055018384238284
102 0.0535933345430514
107 0.0528544888462401
112 0.0514566633408047
117 0.050998643916792
122 0.0498865000219134
127 0.0485373596736503
132 0.0480987114586764
137 0.0471062482073921
142 0.0467358089867099
147 0.0460405672977459
152 0.0452150812746453
157 0.0447844473760813
162 0.0438917464081231
167 0.0436352894051685
172 0.0430492663089187
177 0.0420022360878151
182 0.0418957271277764
187 0.0414960527627264
192 0.0410205669371037
197 0.0405627148860121
};
\addlegendentry{$0.4$}
\addplot [semithick, color4]
table {%
2 0.380677991027797
7 0.206100415816886
12 0.158784711186189
17 0.133872081396965
22 0.11804266375993
27 0.108427441442792
32 0.0995883036983487
37 0.0937251321901938
42 0.0869217043717046
47 0.0835616055860679
52 0.0795272128780594
57 0.0764748812343757
62 0.0741276197129411
67 0.0708775461521847
72 0.0684436221345857
77 0.0667048648648629
82 0.0650114618701662
87 0.0626395994721651
92 0.0614882449943311
97 0.059894900071286
102 0.0591583277810268
107 0.0576522026209081
112 0.056222087641203
117 0.0551798798747448
122 0.0544971585820063
127 0.0532952899899842
132 0.0523184678837822
137 0.0514917035607543
142 0.0504410475660797
147 0.0499218504905624
152 0.0488312560520895
157 0.0487207113966007
162 0.0478475824701821
167 0.0472734056541077
172 0.0469564675066559
177 0.0462946316007029
182 0.0451837309142333
187 0.0452355407052826
192 0.0445177973196509
197 0.0444201548567342
};
\addlegendentry{$0.5$}
\addplot [semithick, color5]
table {%
2 0.433636299148118
7 0.225985549824489
12 0.172447495265311
17 0.144681070586931
22 0.12632793463527
27 0.115463678502017
32 0.106451694403317
37 0.0987641736747971
42 0.0933117455130654
47 0.0881434207903023
52 0.0843445861805985
57 0.080534110525131
62 0.0773398629914481
67 0.0753215325190401
72 0.072124598266059
77 0.0711255208503998
82 0.0688369467284372
87 0.0670938625286927
92 0.0648091969143592
97 0.0629362630817367
102 0.0613756787436198
107 0.060683840960506
112 0.0591695941083085
117 0.0580023549353348
122 0.0565148107626681
127 0.0556830652509098
132 0.0549495439955547
137 0.054400483816236
142 0.0534732326070791
147 0.0522616105749817
152 0.0517124786488961
157 0.0507439101907653
162 0.0497086402534696
167 0.0496960008147733
172 0.0490213720151838
177 0.0480429715041637
182 0.0476347463891692
187 0.0471210012154216
192 0.0466980441057292
197 0.0457816027991288
};
\addlegendentry{$0.6$}
\addplot [semithick, color6]
table {%
2 0.487163682898422
7 0.243850010809112
12 0.182452953090441
17 0.153320707372917
22 0.13446636864049
27 0.122334683273012
32 0.110199110683821
37 0.103458231875023
42 0.0968061248449405
47 0.0914416295069852
52 0.0869955094611006
57 0.0826388368909756
62 0.0793895000222028
67 0.0774942410007423
72 0.0752426786101641
77 0.0719482037460259
82 0.0701443533400199
87 0.0681519307568949
92 0.0667755906130238
97 0.0648458053346998
102 0.0628596354027738
107 0.0621167956578063
112 0.0608044626402176
117 0.0596509579628022
122 0.0581825015797601
127 0.0570423973374581
132 0.0561028877028856
137 0.0550216398585983
142 0.0545050807008609
147 0.0534050133909151
152 0.0524981940877294
157 0.0518319376790784
162 0.0507840891062402
167 0.0503211217325758
172 0.0496768452168698
177 0.0490135475188787
182 0.0483948740985698
187 0.0480390635565748
192 0.047331111570468
197 0.0465858728935806
};
\addlegendentry{$0.7$}
\addplot [semithick, white!49.8039215686275!black]
table {%
2 0.54653343327197
7 0.262535745853621
12 0.194698442856922
17 0.162240311975908
22 0.140411684278372
27 0.125618983178303
32 0.11427496277934
37 0.1064153555228
42 0.0996108078157154
47 0.0938585552739791
52 0.0892062589618247
57 0.0851093026617763
62 0.0815468313127123
67 0.0790619926913535
72 0.0753799689558278
77 0.0740172192765319
82 0.0715446791656724
87 0.06947727837971
92 0.0673777536847511
97 0.0649023428729804
102 0.0639330633430231
107 0.0623383936465138
112 0.0614436342853422
117 0.0592495302118234
122 0.0581607992914455
127 0.0570341701669918
132 0.056208955025723
137 0.0558316807433435
142 0.0545926151611861
147 0.0531202841932312
152 0.0528293599729065
157 0.0518697100782002
162 0.0514184456558771
167 0.0505439269564104
172 0.0490746063411924
177 0.0491699929844805
182 0.04880676690839
187 0.0482104074767157
192 0.0474436112869607
197 0.0470785510794809
};
\addlegendentry{$0.8$}
\addplot [semithick, color7]
table {%
2 0.613750906722801
7 0.296794768784515
12 0.217148073885069
17 0.175775625679736
22 0.154172962974122
27 0.13582193036917
32 0.123387977442596
37 0.113404517253807
42 0.105427614917089
47 0.0987866785584995
52 0.0947497935813157
57 0.0881620989361999
62 0.0852655324083613
67 0.0817717354823047
72 0.0787625706221546
77 0.0763420241469649
82 0.0748884664342424
87 0.0712416513099045
92 0.0695026426152343
97 0.067668613428004
102 0.0656955576078834
107 0.0642342639579085
112 0.0627001986243859
117 0.0607802363324548
122 0.0599011451122547
127 0.0587502662103069
132 0.0569347944116558
137 0.0554956009967581
142 0.0551787525792422
147 0.0541338974286076
152 0.0536101390225846
157 0.0528786193230832
162 0.0515342949410858
167 0.0510830074804416
172 0.0501930044854612
177 0.0495440576622198
182 0.0484965527027852
187 0.0476398031976709
192 0.0473944444292519
197 0.0470667539767268
};
\addlegendentry{$0.9$}
\end{axis}

\end{tikzpicture}}   \caption{For different values of $\normphi$, average normalized positionwise distance of elections with $50$ candidates and between $10$ and $200$ voters sampled from the normalized Mallows model from the frequency matrix of the respective distribution.} \label{fig:ovMallows}
    \end{minipage}
\end{figure}

\section{Validation of Location Framework} \label{se:validation}
Again, in this section, distances and dispersion parameters are always normalized.
To validate whether our approach from \Cref{sub:framework} correctly identifies the ``nature'' of an election, we test its capabilities to find for a given election the dispersion parameter of the closest Mallows distribution with a single central vote. 
That is, given an election, we computed for all $\phi\in \{0,0.001,\dots, 1\}$ its distance to $\calD_\Mallows^{v,\phi}$ and returned the minimizing $\phi$ value.
We compare the computed value to the maximum-likelihood estimator for the dispersion parameter of the underlying Mallows distribution computed from the Kemeny consensus ranking~\citep{MM09} in two experiments. 

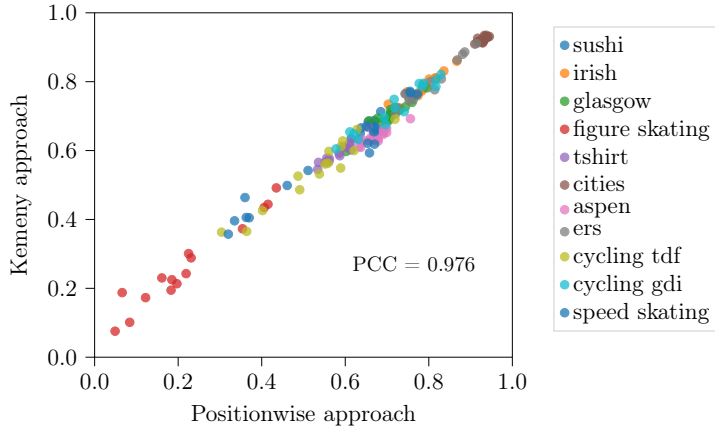
\begin{figure}
    \centering
        \resizebox{0.6\textwidth}{!}{
\begin{tikzpicture}

\definecolor{color0}{rgb}{0.12156862745098,0.466666666666667,0.705882352941177}
\definecolor{color1}{rgb}{1,0.498039215686275,0.0549019607843137}
\definecolor{color2}{rgb}{0.172549019607843,0.627450980392157,0.172549019607843}
\definecolor{color3}{rgb}{0.83921568627451,0.152941176470588,0.156862745098039}
\definecolor{color4}{rgb}{0.580392156862745,0.403921568627451,0.741176470588235}
\definecolor{color5}{rgb}{0.549019607843137,0.337254901960784,0.294117647058824}
\definecolor{color6}{rgb}{0.890196078431372,0.466666666666667,0.76078431372549}
\definecolor{color7}{rgb}{0.737254901960784,0.741176470588235,0.133333333333333}
\definecolor{color8}{rgb}{0.0901960784313725,0.745098039215686,0.811764705882353}

\begin{axis}[
legend cell align={left},
legend style={fill opacity=0.8, draw opacity=1, text opacity=1, at={(1.5,0.96)}, draw=white!80!black},
tick align=outside,
tick pos=left,
x grid style={white!69.0196078431373!black},
xlabel={Positionwise approach},
xmin=0, xmax=1,
xtick style={color=black},
xtick={0,0.2,0.4,0.6,0.8,1},
xticklabels={0.0,0.2,0.4,0.6,0.8,1.0},
y grid style={white!69.0196078431373!black},
ylabel={Kemeny approach},
ymin=0, ymax=1,
ytick style={color=black},
ytick={0,0.2,0.4,0.6,0.8,1},
yticklabels={0.0,0.2,0.4,0.6,0.8,1.0}
]
\addplot [draw=color0, fill=color0, mark=*, only marks, opacity=0.75]
table{%
x  y
0.614 0.634666666667103
0.679 0.663111111110805
0.668 0.649333333333147
0.708 0.705777777777759
0.677 0.655555555555553
0.708 0.697333333333354
0.655 0.683555555555371
0.636 0.649777777778021
0.638 0.665777777777786
0.661 0.677777777777664
0.67 0.683111111111462
0.67 0.667555555555509
0.685 0.712888888888448
0.705 0.692444444444649
0.67 0.687555555555422
};
\addlegendentry{sushi}
\addplot [draw=color1, fill=color1, mark=*, only marks, opacity=0.75]
table{%
x  y
0.722 0.722222222222255
0.748 0.770222222222267
0.867 0.859555555555327
0.791 0.786666666666802
0.783 0.772888888888805
0.817 0.809777777777725
0.786 0.791555555555703
0.786 0.779555555555889
0.816 0.811555555555787
0.703 0.734666666666592
0.774 0.75911111111106
0.731 0.723111111110957
0.801 0.807555555555408
0.836 0.831111111110843
0.765 0.76088888888858
};
\addlegendentry{irish}
\addplot [draw=color2, fill=color2, mark=*, only marks, opacity=0.75]
table{%
x  y
0.678 0.694222222222189
0.699 0.686666666667049
0.76 0.739999999999991
0.744 0.727111111111554
0.696 0.688444444444092
0.708 0.71911111111121
0.602 0.59733333333363
0.689 0.685333333333178
0.719 0.724888888888758
0.796 0.782222222222364
0.708 0.704444444444155
0.73 0.716888888888733
0.711 0.712000000000067
0.657 0.686666666667049
0.704 0.687111111110782
};
\addlegendentry{glasgow}
\addplot [draw=color3, fill=color3, mark=*, only marks, opacity=0.75]
table{%
x  y
0.219 0.242666666666413
0.406 0.434666666666843
0.354 0.372888888889041
0.197 0.213333333333299
0.415 0.4440000000003
0.435 0.491555555555709
0.084 0.101333333333229
0.185 0.224888888888634
0.122 0.172888888888521
0.161 0.230222222222184
0.231 0.288444444444214
0.066 0.18755555555595
0.183 0.19466666666679
0.049 0.0755555555553971
0.225 0.30088888888894
};
\addlegendentry{figure skating}
\addplot [draw=color4, fill=color4, mark=*, only marks, opacity=0.75]
table{%
x  y
0.609 0.618222222222217
0.596 0.612
0.587 0.584888888888879
0.614 0.601777777777837
0.633 0.631555555555518
0.534 0.545333333333288
0.536 0.565777777777361
0.655 0.667111111110675
0.612 0.621333333333005
0.552 0.565333333333443
0.585 0.592000000000001
0.602 0.609333333332932
0.562 0.576444444444791
0.556 0.578222222222538
0.561 0.568444444444071
};
\addlegendentry{tshirt}
\addplot [draw=color5, fill=color5, mark=*, only marks, opacity=0.75]
table{%
x  y
0.934 0.930222222222148
0.945 0.931111111111006
0.931 0.91955555555556
0.941 0.930222222222148
0.929 0.912444444444102
0.91 0.90933333333325
0.917 0.915999999999745
0.923 0.920888888888603
0.917 0.926222222222143
0.928 0.91466666666701
0.938 0.934666666666325
0.925 0.915555555555746
0.932 0.919111111111051
0.933 0.934222222222232
0.94 0.928444444444611
};
\addlegendentry{cities}
\addplot [draw=color6, fill=color6, mark=*, only marks, opacity=0.75]
table{%
x  y
0.633 0.611555555555762
0.691 0.652888888888518
0.691 0.659111111111014
0.679 0.627555555555664
0.691 0.647111111111354
0.678 0.634222222222168
0.655 0.629777777777671
0.756 0.692444444444649
0.661 0.636888888888663
0.699 0.664888888889299
0.637 0.609333333332932
0.699 0.652444444444121
0.669 0.630666666666707
0.671 0.63200000000037
0.641 0.629777777777671
};
\addlegendentry{aspen}
\addplot [
  draw=white!49.8039215686275!black,
  fill=white!49.8039215686275!black,
  mark=*,
  only marks,
  opacity=0.75
]
table{%
x  y
0.914 0.91066666666647
0.763 0.74844444444426
0.743 0.764888888888707
0.881 0.879555555555553
0.83 0.808000000000316
0.725 0.725333333333136
0.753 0.751555555555695
0.753 0.744888888888793
0.814 0.776444444444299
0.72 0.705777777777759
0.745 0.764444444444367
0.868 0.86311111111087
0.774 0.76711111111115
0.804 0.799999999999839
0.886 0.886222222222655
};
\addlegendentry{ers}
\addplot [draw=color7, fill=color7, mark=*, only marks, opacity=0.75]
table{%
x  y
0.628 0.661333333333183
0.304 0.363111111111383
0.538 0.532000000000057
0.59 0.612
0.491 0.486222222222289
0.402 0.425777777777985
0.589 0.549333333333014
0.559 0.563555555555895
0.364 0.365333333333565
0.594 0.628000000000255
0.719 0.691111111111294
0.487 0.525777777777426
0.561 0.597777777777786
0.552 0.560888888888534
0.621 0.599999999999821
};
\addlegendentry{cycling tdf}
\addplot [draw=color8, fill=color8, mark=*, only marks, opacity=0.75]
table{%
x  y
0.578 0.604888888888785
0.817 0.785777777777864
0.693 0.696888888888593
0.717 0.748888888889163
0.697 0.67511111111098
0.631 0.632444444444761
0.788 0.79555555555543
0.621 0.648444444444384
0.817 0.804888888888831
0.785 0.786222222222341
0.777 0.794666666666483
0.719 0.723555555555668
0.741 0.712888888888448
0.611 0.654222222222398
0.829 0.820888888888781
};
\addlegendentry{cycling gdi}
\addplot [draw=color0, fill=color0, mark=*, only marks, opacity=0.75]
table{%
x  y
0.32 0.357333333333194
0.658 0.593333333333361
0.37 0.404888888888969
0.651 0.670666666666795
0.756 0.768000000000259
0.772 0.763111111111069
0.669 0.654666666667
0.654 0.621333333333005
0.511 0.542222222222556
0.363 0.405777777777692
0.67 0.618666666666716
0.461 0.498666666666629
0.36 0.463555555555544
0.755 0.771999999999931
0.335 0.396000000000099
};
\addlegendentry{speed skating}
\draw (axis cs:0.6,0.25) node[
  scale=0.9,
  anchor=base west,
  text=black,
  rotate=0.0
]{PCC = 0.976};
\end{axis}

\end{tikzpicture}}   \caption{Correlation between the predicted dispersion parameter of our real-world data by our (positionwise) approach and by the maximum-likelihood Kemeny approach. } \label{fig:corr}
\end{figure}

We start by sampling  for $\phi\in \{0,0.05,\dots,1\}$ an election with $10$ candidates and $100$ voters from Mallows model with dispersion parameter $\phi$ and computed estimates for the dispersion parameter based on our and the Kemeny approach. 
For all elections, the returned estimates differ by at most $0.01$. 
Thus, the dispersion parameter returned by our approach is always very close to the maximum-likelihood estimate. 
However, the estimated dispersion parameter might deviate a bit from the originally used dispersion parameter: On average, the absolute difference between the dispersion parameter returned by our approach and the underlying dispersion parameter is $0.0179$ with the maximum difference being $0.069$; for Kemeny, the average is $0.016$ and the maximum is $0.082$. 
While it might seem surprising that the estimated dispersion parameter is different from the underlying one, recall from \Cref{sec:variance} that elections sampled from a vote distribution typically have non-zero distance from the distribution's frequency matrix. 
To illustrate this idea, we can think of a Mallows distribution as a normal distribution placed in the space of elections with the dispersion parameter being its mean (in particular as soon as the dispersion parameter is greater zero, all elections have a non-zero probability of being sampled).
So multiple Mallows distributions for different dispersion parameters translate to multiple partly overlapping normal distributions and it might as well happen that an election sampled from a Mallows distribution with one dispersion parameter is in fact closer to the mean of the Mallows distribution with a different dispersion parameter. 

We also repeated the above experiment to measure the capabilities of our approach to estimate the parameters of a mixed Mallows distribution. For each pair of $p\in \{0.1,0.2,0.3,0.4,0.5\}$ and $\phi \in \{0.05,0.1,\dots,0.95\}$ we sampled an election with $10$ candidates and $100$ voters from  $p$-$\calD_\Mallows^{v, \phi, \phi}$ (i.e., we sample a vote from the Mallows distribution with dispersion parameter $\phi$ and subsequently flip the sampled vote with probability $p$). Subsequently we computed for which value of $p\in \{0,0.05,\dots, 0.5\}$ and $\phi\in \{0,0.05,\dots, 1\}$ the distance of the sampled election is closest to the frequency matrix of the induced distribution $p$-$\calD_\Mallows^{v, \phi, \phi}$  (the Kemeny consensus ranking can no longer be used here to provide a maximum-likelihood estimate). 
The average difference between the estimated and underlying dispersion parameter is $0.028$ and the average difference between the estimated and underlying flipping probability is $0.068$. 
For $63$ out of the $95$ elections is the difference between the estimated and underlying dispersion parameter and weight smaller equal $0.05$ (which is the smallest non-zero difference). 
The error of the estimated dispersion parameter here is around three times larger than for Mallows elections with a single central vote. 
This can be explained by the fact that for mixtures of Mallows distributions, the ``overlap'' between different distributions is even larger; in fact, some paramterizations even result in the same distributions (e.g., for $\phi=1$ all flipping probabilities result in the same distribution). 

Second, while producing good estimates for elections that have been sampled from a Mallows distribution is a good sanity check, we are ultimately interested in computing to which distribution (unknown) real-world elections are closest. 
To do so, we again compare the estimated dispersion parameter for a Mallows distribution with a single central vote computed by our approach with the one estimated via the Kemeny consensus ranking (as described in the beginning of this section); however, this time instead of considering elections sampled from Mallows model, we examine $165$ real-world elections used by \citet{boe-bre-fal-nie-szu:t:compass} (see the data part of \Cref{sub:framework} for details on the dataset). 
The estimated dispersion parameters returned by both methods are highly correlated with a Pearson correlation coefficient of $0.976$ and an average difference of $0.017$, median difference of $0.0105$, and maximum difference of $0.197$, indicating the power of our approach.
Interestingly, the correlation is particularly strong for larger dispersion parameters (see Figure \Cref{fig:corr} for a plot showing the correlation between the two approaches). 
Together with the estimated normalized dispersion parameter, both approaches also return the central order $v$ of the closest Mallows model, which are typically again quite similar: the average swap distance between the two estimators is $2.81$ out of $45$ possible swaps.

\end{document}